\documentclass[runningheads]{llncs}
\usepackage[T1]{fontenc}
\usepackage{graphicx}
\usepackage{hyperref}
\hypersetup{
    colorlinks=true,
    linkcolor=blue,
    citecolor= blue,
    urlcolor=cyan
    }
\usepackage{color}

\usepackage[table,x11names]{xcolor}
\usepackage{amsmath}
\usepackage{amsfonts}
\usepackage{mathtools}
\usepackage{amssymb}
\usepackage{enumerate}
\usepackage{longtable}
\usepackage{float}
\usepackage{orcidlink}
\usepackage{multirow} 
\usepackage{hhline} 

\usepackage[
n,
operators,
advantage,
sets,
adversary,
landau,
probability,
notions,	
logic,
ff,
mm,
primitives,
events,
complexity,
asymptotics,
keys]{cryptocode}
\newcommand{\splitatcommas}[1]{%
	\begingroup
	\begingroup\lccode`~=`, \lowercase{\endgroup
		\edef~{\mathchar\the\mathcode`, \penalty0 \noexpand\hspace{0pt plus 1em}}%
	}\mathcode`,="8000 #1%
	\endgroup
}
\newcommand*{\Comb}[2]{{}^{#1}C_{#2}}%

\newtheorem{fact}{\bf Fact}

\begin{document}
	
\mainmatter          

\title{On the Construction of Near-MDS Matrices}

\author{Kishan Chand Gupta\inst{1} \and
Sumit Kumar Pandey\inst{2} \and
Susanta Samanta\inst{3} \orcidlink{0000-0003-4643-5117}}

\authorrunning{K. C. Gupta et al.}
	
\institute{Applied Statistics Unit, Indian Statistical Institute, \\ 203, B.T. Road, Kolkata-700108, INDIA. \\ \email{kishan@isical.ac.in} \and Computer Science and Engineering, Indian Institute of Technology Jammu,\\ Jagti, PO Nagrota, Jammu-181221, INDIA. \\ \email{emailpandey@gmail.com} 
\and Applied Statistics Unit, Indian Statistical Institute, \\ 203, B.T. Road, Kolkata-700108, INDIA. \\ \email{susantas\_r@isical.ac.in}}
\maketitle
\begin{abstract}
The optimal branch number of MDS matrices makes them a preferred choice for designing diffusion layers in many block ciphers and hash functions. However, in lightweight cryptography, Near-MDS (NMDS) matrices with sub-optimal branch numbers offer a better balance between security and efficiency as a diffusion layer, compared to MDS matrices.
In this paper, we study NMDS matrices, exploring their construction in both recursive and nonrecursive settings. 
We provide several theoretical results and explore the hardware efficiency of the construction of NMDS matrices. Additionally, we make comparisons between the results of NMDS and MDS matrices whenever possible.
For the recursive approach, we study the DLS matrices and provide some theoretical results on their use. Some of the results are used to restrict the search space of the DLS matrices.
%
We also show that over a field of characteristic 2, any sparse matrix of order $n\geq 4$ with fixed XOR value of 1 cannot be an NMDS when raised to a power of $k\leq n$.
Following that, we use the generalized DLS (GDLS) matrices to provide some lightweight recursive NMDS matrices of several orders that perform better than the existing matrices in terms of hardware cost or the number of iterations.
For the nonrecursive construction of NMDS matrices, we study various structures, such as circulant and left-circulant matrices, and their generalizations: Toeplitz and Hankel matrices. In addition, we prove that Toeplitz matrices of order $n>4$ cannot be simultaneously NMDS and involutory over a field of characteristic 2.
Finally, we use GDLS matrices to provide some lightweight NMDS matrices that can be computed in one clock cycle. The proposed nonrecursive NMDS matrices of orders 4, 5, 6, 7, and 8 can be implemented with 24, 50, 65, 96, and 108 XORs over $\mathbb{F}_{2^4}$, respectively. 

\end{abstract}

\section{Introduction}
Shannon's concepts of confusion and diffusion~\cite{SHANON} are well demonstrated through the design of symmetric key cryptographic primitives. In many cases, the round function of the design uses both non-linear and linear layers to achieve confusion and diffusion, respectively.
The focus of this paper is the construction of linear diffusion layers whose purpose is to maximize the spreading of internal dependencies. Optimal diffusion layers can be achieved by utilizing MDS matrices with maximum branch numbers. An example of this is the MixColumn operation in AES~\cite{AES} which employs a $4\times 4$ MDS matrix. The use of an MDS matrix makes AES robust against differential and linear attacks with a low number of rounds, which is ideal for low-latency applications.

\par The growth of ubiquitous computing such as IoT has highlighted new security needs, driving research in the field of lightweight cryptography. Currently, there is a growing focus on the study of lightweight diffusion matrices. Numerous proposals for constructing lightweight MDS and involutory MDS matrices have been made~\cite{BKL2016,GR15,CYCLICM,GHADA,XORM,PXOR1}. It is important to note that elements of an MDS matrix must be nonzero, resulting in a high hardware implementation cost. To minimize hardware cost, Guo et al.~\cite{PHOTON,LED} proposed a novel design approach of using recursive MDS matrices, which have a low hardware area but require additional clock cycles. Subsequently, researchers have focused on designing recursive MDS matrices, producing a significant number of results~\cite{Augot2014,Berger2013,Rec_MDS_2022,GuptaPV17_2,GuptaPV17_1,Sajadieh_Recursive_Diffusion_Layer12,DSI,Recursive_Diffusion_Layer12,Recursive_Diffusion_Layer2014}. Instead of recursive approach, Sajadieh et al.~\cite{SM19} constructed lightweight MDS matrices by the composition of different sparse matrices and the block cipher FUTURE~\cite{FUTURE} use this idea for the MDS matrix in its MixColumn operation.

\par We have gained a thorough understanding of local optimization techniques so far. As a consequence, recent efforts have shifted to addressing the problem on a more fundamental level, viewing it as the well-known Shortest Linear Straight-Line Problem (SLP), which was first used in~\cite{Boyar_2008} for global optimization of a pre-defined linear function. This leads to the construction of more lightweight MDS matrices~\cite{Duval_Leurent_2018,Kranz_Leander_Stoffelen_Wiemer_2017,Li_Sun_Li_Wei_Hu_2019,Yang_Zeng_Wang_2021}.

\par However, the trade-off between security and efficiency may not be optimal with MDS and recursive MDS matrices. Near-MDS (NMDS) matrices are characterized by having sub-optimal branch numbers, leading to a slower diffusion speed and smaller minimum active Sboxes per round compared to ciphers using MDS matrices.
However, studies such as~\cite{Alfarano_2018,MIDORI} have demonstrated that NMDS matrices, when used with a well-chosen permutation, can improve security against differential and linear cryptanalysis. Thus, NMDS matrices offer better trade-offs between security and efficiency when used in the diffusion layer of lightweight cryptography.
Some recent lightweight block ciphers, such as PRIDE~\cite{PRIDE}, Midori~\cite{MIDORI}, MANTIS~\cite{SKINNY}, FIDES~\cite{Fides} and PRINCE~\cite{PRINCE} have utilized NMDS matrices. The importance of lightweight symmetric key primitives with low power, low energy, or low latency is growing, and NMDS matrices are commonly employed in the construction of lightweight block ciphers. However, NMDS or recursive NMDS matrices have been less studied in the literature up to now. This inspires us to introduce new results on NMDS matrices and in this work, we look at both recursive NMDS matrices and NMDS matrices that can be constructed using one clock cycle.
%
%
\\ \\
\textbf{Contributions.~}
The primary focus of this paper is the study of NMDS matrices, examining their construction using both recursive and nonrecursive approaches. We present various theoretical results and also discuss the hardware efficiency of the construction of NMDS matrices. We also compare the results of NMDS with MDS matrices whenever possible.

\par The exhaustive search using a naive way of finding a higher order recursive NMDS matrix using DLS matrices is impractical. In this regard, we present some theoretical results that are used to restrict the search space to a small domain. Following that, we observe that there are no $k$-NMDS (see Definition~\ref{Deef_Recursive_NMDS}) DLS matrices over $\mathbb{F}_{2^4}$ for orders $n=5,6,7,8$ and $k\in \set{n-1,n}$ if the fixed XOR (say $\mathcal{K}$) of the matrix is taken to be less than $\ceil{\frac{n}{2}}$.
\par In~\cite[Theorem~$3$]{Rec_MDS_2022}, the authors proved that for a DLS matrix $M$ of order $n\geq 2$, $M^k$ is not MDS for $k < n-1$. In Theorem~\ref{th_DLS_power_need_nMDS}, we demonstrate that for a DLS matrix $M$ of order $n\geq 2$, $M^k$ is not NMDS for $k<n-2$.
Furthermore, we discuss the impact of the permutation $\rho$ in a DLS matrix $DLS(\rho; D_1, D_2)$ for the construction of recursive NMDS matrices. More specifically, in Corollary~\ref{corollary_min_power_in_PD1+D2_if_P_not_full_cycle}, we show that if $\rho$ is not an $n$-cycle, a DLS matrix of order $n\geq 2$, cannot be $k$-NMDS for $k\leq n$.

\par In~\cite[Theorem~2]{OUR_DSI}, it was proved that no sparse matrix of order $n\geq 3$ with $\mathcal{K}=1$ can yield an MDS matrix when raised to power $n$. In Theorem~\ref{Th_1_XOR_matrix_not_NMDS}, we show that over a field of characteristic $2$, there is no sparse matrix of order $n\geq 4$ with $\mathcal{K}=1$ can yield NMDS matrix when raised to power $k\leq n$. As a consequence, it is shown that for a $k$-NMDS matrix of orders $4$ with $k\leq 4$, the lowest XOR count is $2r$ over the field $\mathbb{F}_{2^r}$. 

\par Using GDLS matrices, we also introduce some lightweight recursive NMDS matrices of orders $4, 5, 6$ and $7$ that can be implemented with $8, 13, 13$, and $18$ XORs over $\FF_{2^4}$, respectively. Besides searching over $\FF_{2^4}$, we also provide some efficient $k$-NMDS GDLS matrices with $k\in \set{n-1,n}$ over $GL(8,\FF_{2})$ for various orders $n$. Table~\ref{table_comparison_k-NMDS} compares our results with the known results.

\par We examine different structures for the nonrecursive construction of NMDS matrices, including circulant and left-circulant matrices, as well as their generalizations such as Toeplitz and Hankel matrices. Proposition 3 in~\cite{Li_Wang_2017} demonstrates that circulant matrices of order $n>4$ cannot be both NMDS and involutory over $\mathbb{F}_{2^r}$. In Theorem~\ref{Th_Toeplitz_NMDS_involutory}, we prove that this result also holds for Toeplitz matrices. 

\par According to \cite[Lemma~5]{GR15}, it has been established that no orthogonal circulant matrix of order $2^n$ over the field $\mathbb{F}_{2^r}$ can be MDS for $n \geq 2$. However, in Remark~\ref{Remark_circulant_NMDS_orthogonal}, we see that NMDS circulant orthogonal matrices of any order may exist over the field $\mathbb{F}_{2^r}$. Similar result shows up for left-circulant matrices. Table~\ref{Table_Comparision_MDS_NMDS} is provided to compare the involutory and orthogonal properties of MDS and NMDS matrices constructed from the circulant, left-circulant, Toeplitz, and Hankel families.

\par For Hadamard, circulant, or left-circulant NMDS matrices of order $n$, we must have $\mathcal{K}\geq n(n-2)$, which results in a high implementation cost for constructing NMDS matrices from such matrices. To address this issue, we use composition of different GDLS matrices to construct nonrecursive NMDS matrices. Our proposed nonrecursive NMDS matrices for orders 4, 5, 6, 7, and 8 can be implemented using 24, 50, 65, 96, and 108 XORs over $\mathbb{F}_{2^4}$, respectively. Table~\ref{table_comparison_nonrec_NMDS} compares our results for nonrecursive NMDS matrices with the known results. 
\par To compare with MDS matrices, we examine certain well-known results of MDS matrices and apply them to NMDS matrices. For example, in Remark~\ref{Remark_submatrix_NMDS_is_MDS} and Remark~\ref{Remark_NMDS_is_singular}, we observe that a submatrix of an NMDS matrix may not necessarily be an NMDS matrix, and an NMDS matrix can also be singular. However, we prove that the inverse of a nonsingular NMDS matrix is also an NMDS matrix. In Corollary~\ref{corollary_D1AD2_nMDS}, we prove that if $M$ is an NMDS matrix and $D_1$ and $D_2$ are two nonsingular diagonal matrices, then $D_1MD_2$ is also an NMDS matrix.
\\ \\
\textbf{Outline.~}The remaining sections of this paper are structured as follows. In Section~\ref{Sec_Definitions}, definitions and preliminaries are presented. Section~\ref{Section_Rec_NMDS} discusses DLS matrices for the construction of recursive NMDS matrices. Some lightweight recursive NMDS matrices of various orders, using GDLS matrices, are proposed in Section~\ref{Section_Rec_NMDS_GDLS}. Section~\ref{Section_Single_clock_NMDS} discusses circulant, left-circulant, Toeplitz and Hankel matrices for the construction of nonrecursive NMDS matrices. Section~\ref{Section_nonrecursive_NMDS_GDLS} provides some lightweight nonrecursive NMDS matrices constructed from GDLS matrices. Section~\ref{Section_NMDS_conclusion_future_work} concludes the paper and discusses some possible future work.

\section{Definition and Preliminaries}\label{Sec_Definitions}
Before discussing NMDS matrices in depth, let us recall some basic notations for finite fields, their representations, and the construction of matrices.

Let $\mathbb{F}_{2^r}$ be the finite field of $2^r$ elements and $\mathbb{F}_{2^r}^n$ be the set of vectors of length $n$ with entries from the finite field $\mathbb{F}_{2^r}$. Let $\mathbb{F}_{2^r}^*$ be the multiplicative group of $\mathbb{F}_{2^r}$. 
It is a well-established fact that elements of a finite field with characteristic $2$ can be represented as vectors with coefficients in $\mathbb{F}_2$. In other words, there exists a vector space isomorphism from $\mathbb{F}_{2^r}$ to $\mathbb{F}_2^r$ defined by $x=(x_1\alpha_1+x_2\alpha_2+ \cdots +x_r\alpha_r) \rightarrow (x_1,x_2, \ldots,x_r)$, where $\{\alpha_1,\alpha_2,\ldots,\alpha_r\}$ is a basis of $\mathbb{F}_{2^r}$.
If $\alpha$ is a primitive element of $\mathbb{F}_{2^r}$, every nonzero element of $\mathbb{F}_{2^r}$ can be expressed as a power of $\alpha$. Hence, $\mathbb{F}_{2^r}^*=\set{1,\alpha,\alpha
^2,\alpha^3,\ldots,\alpha^{2^r-1}}$. In favor of a more compact notation, we also use a hexadecimal representation of a finite field $\FF_{2^r}$ throughout the paper. For instance, $\FF_{2^4}/0$x$13$ denotes the finite field $\FF_{2^4}$ constructed by the polynomial $x^4+x+1$. 

A square matrix is a matrix with the same number of rows and columns. An $n\times n$ matrix is known as a matrix of order $n$. 
The ring of $n\times n$ matrices over $\mathbb{F}_{2^r}$ is denoted by $M_n(\mathbb{F}_{2^r})$ and the general linear group consisting of nonsingular $n\times n$ matrices over $\mathbb{F}_{2^r}$ is denoted by $GL(n,\mathbb{F}_{2^r})$.

The diffusion power of a linear transformation, as specified by a matrix, is quantified by its branch numbers~\cite[pages 130--132]{AES}.
\begin{definition}~\cite[page $132$]{AES}
    The differential branch number, $\beta_{d}(M)$, of a matrix $M$ of order $n$ over the finite field $\FF_{2^r}$ is defined as the smallest number of nonzero components in both the input vector $x$ and the output vector $Mx$, as we consider all nonzero $x$ in $(\FF_{2^r})^n$ i.e.
    \begin{center}
        $\beta_{d}(M)=min_{x\neq \mathbf{0}}(w(x)+w(Mx))$,
    \end{center}
    where $w(x)$ represents the number of nonzero components in the vector $x$.
\end{definition}
\begin{definition}~\cite[page $132$]{AES}
    The linear branch number, $\beta_{l}(M)$, of a matrix $M$ of order $n$ over the finite field $\FF_{2^r}$ is defined as the smallest number of nonzero components in both the input vector $x$ and the output vector $M^Tx$, as we consider all nonzero $x$ in $(\FF_{2^r})^n$ i.e.
    \begin{center}
        $\beta_{l}(M)=min_{x\neq \mathbf{0}}(w(x)+w(M^Tx))$,
    \end{center}
    where $w(x)$ represents the number of nonzero components in the vector $x$.
\end{definition}

\begin{remark}~\cite[page $144$]{AES}\label{Remark_min_dis_is_branch_number}
    The differential branch number $\beta_{d}(M)$ of a matrix $M$ is equal to the minimum distance of the linear code $C$ generated by the matrix $[I~|~M]$. Furthermore, $\beta_{l}(M)$ is equivalent to the minimum distance of the dual code of $C$.
\end{remark}

\begin{remark}~\cite[page $132$]{AES}\label{Remark_branch_number}
    It is important to note that the maximum value for both $\beta_{d}(M)$ and $\beta_{l}(M)$ is $n + 1$. While $\beta_{d}(M)$ and $\beta_{l}(M)$ are not always equal, a matrix with the highest possible differential or linear branch number will have the same value for both.
\end{remark}
The singleton bound tells us that $d\leq n-k+1$ for a $[n,k,d]$ code $C$. We call $C$ an MDS (maximum distance separable) code if $d=n-k+1$. In this context, we state an important theorem from coding theory.

\begin{theorem}~\cite[page 321]{FJ77}
    An $[n, k, d]$ code $C$ with generator matrix $G = [ I ~|~ M ]$, where $M$ is a $k \times ( n - k )$ matrix, is MDS if and only if every square submatrix (formed from any $i$ rows and any $i$ columns, for any $i = 1, 2 , \ldots, min \{k, n - k \}$) of $M$ is nonsingular.
\end{theorem}
Another way to define an MDS matrix is as follows.

\begin{fact}
    A square matrix $M$ is an MDS matrix if and only if every square submatrices of $M$ are nonsingular. 
\end{fact}
In this paper, we discuss the diffusion matrices with the highest branch numbers among non-MDS matrices.

\begin{definition}\cite{Li_Wang_2017}
    A matrix $M$ of order $n$ is called a near-MDS (in short NMDS) matrix if $\beta_{d}(M)=~\beta_{l}(M)= n$.
\end{definition}

In~\cite{DL95_nMDS}, an $[n, k, d]$ NMDS code $C$ is defined by the conditions $d = n-k$ and $d'= k$, where $d'$ is the minimum distance of the dual code of $C$. Thus, by Remark~\ref{Remark_min_dis_is_branch_number}, for a matrix $M$ of order $n$ with $\beta_{d}(M)=\beta_{l}(M)=n$, the matrix $[I~|~M]$ is exactly a generator matrix of a $[2n, n, n]$ NMDS code. Thus, the following characterization of an NMDS matrix is obtained.

\begin{lemma}\cite{Li_Wang_2017,Viswanath2006_nMMDS}\label{Lemma_nMDS_matrix_characterization}
    Let $M$ be a non-MDS matrix of order $n$, where $n$ is a positive integer with $n \geq 2$. Then $M$ is NMDS if and only if for any $1 \leq g \leq n-1$ each $g \times (g + 1)$ and $(g + 1) \times g$ submatrix of $M$ has at least one $g \times g$ nonsingular submatrix.
\end{lemma}

In Lemma~\ref{Lemma_nMDS_matrix_characterization}, if we assume $g=1$, we can deduce that there is at most one 0 in each row and each column of an NMDS matrix. Hence, we have the following corollary.

\begin{corollary}\label{corollary_min_nonzero_nMDS}
    An NMDS matrix  $M$ of order $n$ must contain at least $n^2-n$ nonzero elements.
\end{corollary}

\begin{remark}\label{Remark_submatrix_NMDS_is_MDS}
    We know that every square submatrices of an MDS matrix are MDS. However, a square submatrix of an NMDS matrix may not be an NMDS matrix. For example, consider the matrix 
    \begin{equation*}
        M=\begin{bmatrix}
            0 & \alpha & 1 & \alpha + 1 \\
            \alpha + 1 & 0 & \alpha & 1 \\
            1 & \alpha + 1 & 0 & \alpha \\
            \alpha & 1 & \alpha + 1 & 0
        \end{bmatrix}
    \end{equation*}
    over $\FF_{2^4}/0$x$13$, where $\alpha$ is a root of $x^4+x+1$. Then it can be checked that $M$ is an NMDS matrix. However, the $2\times 2$ submatrix 
    $\begin{bmatrix}
    1 & \alpha + 1 \\
    \alpha & 1
    \end{bmatrix}$
    of $M$ is an MDS matrix.
\end{remark}

\begin{corollary}\label{transpose_is_near-MDS}
    If $A$ is an NMDS matrix, then $A^{T}$ is also an NMDS matrix.
\end{corollary}

\begin{remark}\label{Remark_NMDS_is_singular}
    We know that an MDS matrix cannot be singular, whereas an NMDS matrix can be singular. For example, the NMDS matrix $M$ in Remark~\ref{Remark_submatrix_NMDS_is_MDS} is singular.
\end{remark}

\begin{lemma}
    For a nonsingular NMDS matrix $M$, its inverse $M^{-1}$ is also an NMDS matrix.
\end{lemma}
\begin{proof}
    Let $G=[I~|~M]$ be a generator matrix of an NMDS code. Elementary row operation change $G=[I~|~M]$ to $G'=[M^{-1}~|~I]$. Since elementary row operations do not change the code, $G'$ is also a generator matrix of the NMDS code. So the code defined by $[I~|~M^{-1}]$ has the same minimal distance. Therefore, $M^{-1}$ is an NMDS matrix.
\end{proof}

\begin{definition}\label{Deef_Recursive_NMDS}
    Let $k$ be a positive integer. A matrix $B$ is said to be recursive NMDS or $k$-NMDS if the matrix $M=B^k$ is NMDS. If $B$ is $k$-NMDS then we say $B$ yields an NMDS matrix.
\end{definition}

\begin{example}\label{Example_22-MDS}
    For example, the matrix	
    $$B=
    \left[ \begin{array}{rrrr}
    0 & 1 & 0 & 0 \\
    0 & 0 & 1 & 0 \\
    0 & 0 & 0 & 1 \\
    1 & \alpha & 0 & 0
    \end{array} \right]
    $$
    is $10$-NMDS, where $\alpha$ is a primitive element of the field $\FF_{2^4}$ and a root of $x^4+x+1$.
\end{example}

\begin{definition}
    A matrix $D$ of order $n$ is said to be diagonal if $(D)_{i,j}=0$ for $i\neq j.$ By setting $d_i=(D)_{i,i}$, we denote the diagonal matrix $D$ as $diag(d_1, d_2,\ldots ,d_n)$.
\end{definition}

It is obvious to see that the determinant of $D$ is $det(D) =\prod_{i=1}^{n}d_i$. Hence, the diagonal matrix $D$ is nonsingular if and only if $d_i\neq 0$ for $1\leq i \leq n$.

\noindent It is worth exploring whether the NMDS property of a matrix remains invariant under the elementary row operation of multiplying a row or column by a nonzero scalar. Therefore, we have the following lemma.

\begin{lemma}\label{lemma_D1AD2_nMDS}
	Let $M$ be an NMDS matrix over $\mathbb{F}_{2^r}$, then $M'$, obtained by multiplying a row (column) of $M$ by any $c \in \mathbb{F}_{2^r}^*$ will also be an NMDS.
\end{lemma}
\begin{proof}
	Take $B'$ be an arbitrary $g\times (g+1)$  ($(g+1)\times g$) submatrix of $M'$. Suppose $B$ is the corresponding submatrix of $M$. Since $M$ is NMDS matrix, $B$ must have a nonsingular $g\times g$ submatrix $I$. Let $I'$ be the corresponding submatrix of $B'$. If the submatrix $I'$ contains the row (column) in which $c$ has multiplied, then $det(I')=c\cdot det(I)$ otherwise $det(I')=det(I)$. Thus, $B'$ contains a nonsingular $g\times g$ submatrix $I'$. Therefore, by Lemma~\ref{Lemma_nMDS_matrix_characterization}, $M'$ is also an NMDS matrix.\qed
\end{proof}
Let $D=diag(c_1,c_2,\ldots ,c_{n})$ be a diagonal matrix. Then by the multiplication $DM$ ($MD$) it means multiply the $i$-th row ($i$-th column) of $M$ by $c_i$ for $1\leq i \leq n.$ Hence, we can generalize the Lemma~\ref{lemma_D1AD2_nMDS} as follows.

\begin{corollary}\label{corollary_D1AD2_nMDS}
	Let $M$ be an NMDS matrix, then for any nonsingular diagonal matrices $D_1$ and $D_2$, $D_1MD_2$ will also be an NMDS matrix.
\end{corollary}

\begin{corollary}\label{corollary_DAD^(-1)_MDS}
	Let $B$ be a recursive NMDS matrix, then for any nonsingular diagonal matrix $D$, $DBD^{-1}$ will also be a recursive NMDS matrix.
\end{corollary}
\begin{proof}
	Suppose $D$ is a nonsingular diagonal matrix and $B$ is $k$-NMDS i.e. $B^k$ is an NMDS matrix. Then we have 
	\begin{align*}
		(DBD^{-1})^k &= \underbrace{DBD^{-1}\cdot DBD^{-1}\cdot~ \ldots ~\cdot DBD^{-1}}_{k-\text{times}}= DB^kD^{-1}
	\end{align*}
Now since $D$ is a nonsingular diagonal matrix and $B^k$ is an NMDS matrix, from Corollary~\ref{corollary_D1AD2_nMDS}, we can say that $DB^kD^{-1}$ is again an NMDS matrix. Hence, $DBD^{-1}$ is a recursive NMDS matrix. More specifically $DBD^{-1}$ is $k$-NMDS. \qed
\end{proof}

\begin{definition}
	Let $\rho$ be an element of the symmetric group $S_n$ (set of all permutations over the set $\set{1,2,\ldots ,n}$~). Then by $\rho=~ [i_1,~i_2,~i_3,~\ldots~,i_n]$, where $1\leq i_j\leq n$ for $j= 1,2,3,\ldots,n$, we mean $\rho=\begin{pmatrix}
	1 & 2 & 3 & \ldots  & n \\
	i_1 & i_2 & i_3 & \ldots &  i_n 
	\end{pmatrix}$
	i.e. $1\rightarrow i_1$, $2 \rightarrow i_2$, $\ldots$, $n\rightarrow i_n$.
\end{definition}
Then the product of two permutations $\rho_1=~ [i_1,~i_2,~i_3,~\ldots~,i_n]$ and $\rho_2=~ [\splitatcommas{j_1,~j_2,~j_3,~\ldots~,j_n}]$ is given by
	$\rho_1\cdot \rho_2=[i_{j_1},~i_{j_2},~i_{j_3},~\ldots~,i_{j_n}]$
and the inverse of a permutation $\rho=~ [i_1,~i_2,~i_3,~\ldots~,i_n]$ is the permutation $\delta=~ [\splitatcommas{j_1,~j_2,~j_3,~\ldots~,j_n}]$ such that $\rho\cdot \delta=\delta\cdot \rho=~[1,2,3,~\ldots~,n]$.

\begin{example}
	For the two permutations $\rho_1=[2,3,4,5,1,6]$ and $\rho_2=[1,4,3,2,6,5]$ over $S_6$, their product is given by
		$$\rho_1\cdot \rho_2=[2,5,4,3,6,1]~\text{and}~\rho_2\cdot \rho_1=[4,3,2,6,1,5].$$
	The inverse of the permutation $\rho_1=[2,3,4,5,1,6]$ is given by $\delta=[5, 1, 2, 3, 4, 6]$.
\end{example}

\begin{definition}
	A permutation matrix $P$ of order $n$ related to a permutation $\rho=~ [\splitatcommas{i_1,~i_2,~i_3,~\ldots,~i_n}]$ is the binary matrix which is obtained from the identity matrix of order $n$  by permuting the rows (columns) according to the permutation $\rho$.
\end{definition}

\noindent In this paper, we will use the row permuted identity matrix to represent permutation matrices. For instance, the permutation matrix $P$ related to the permutation $[4,2,3,1]$ is given by
$$P=~
\begin{bmatrix}
0&~0&~0&~1\\
0&~1&~0&~0\\
0&~0&~1&~0\\
1&~0&~0&~0
\end{bmatrix}.
$$ 
Note that a permutation matrix is invertible and the inverse of $P$ is the transpose of $P$, i.e. $P^{-1}=P^T$. The product of two permutation matrices is a permutation matrix.

\begin{definition}
	Let $\rho$ be an element of the symmetric group $S_n$. Then $\rho$ is called a $k$ length cycle or $k$-cycle, written $(j_1~j_2~j_3~\ldots~j_k)$, if $\rho=\begin{pmatrix}
	j_1 & j_2 & j_3 & \ldots  & j_k \\
	j_2 & j_3 & j_4 & \ldots &  j_1 
	\end{pmatrix}$
	i.e. $j_1\rightarrow j_2$, $j_2 \rightarrow j_3$, $\ldots$, $j_k\rightarrow j_1$.
\end{definition}

For example, the permutation $\rho= [3,2,4,1,5]$ can be written as $(1~3~4)$. So $\rho$ is a 3-cycle in $S_5$.

\noindent The branch number of a matrix is invariant under row (column) permutation. So, from~\cite{MDS_Survey}, we get the following result.

\begin{lemma}\cite{MDS_Survey}\label{Lemma_same_branch_perm}
	For any permutation matrices $P$ and $Q$, the branch numbers of the two matrices $M$ and $PMQ$ are same.
\end{lemma}

\begin{corollary}\label{corollary_PAP^(-1)_MDS}
	Let $B$ be a recursive NMDS matrix, then for any permutation matrix $P$, $PBP^{-1}$ will also be a recursive NMDS matrix.
\end{corollary}

\begin{definition}\label{Def_perm_diag_similar}
	We will call a matrix $M_1$ to be diagonal (permutation) similar to a matrix $M_2$ if $M_1=DM_2D^{-1}$ ($M_1=PM_2P^{-1}$) for some nonsingular diagonal matrix $D$ (permutation matrix $P$).
\end{definition}

\begin{fact}\label{Fact_perm_similar_mds_nmds}
	Diagonal (permutation) similar of a $k$-NMDS matrix is again a $k$-NMDS matrix.
\end{fact}

\subsubsection{XOR count:}
The efficiency of hardware implementation in a given operation is typically assessed by the amount of area required. It is worth noting that the diffusion matrix can only be implemented using XOR gates, which leads to the following definition.

\begin{definition}~\cite{Lukas_2019}
    The direct XOR count (d-XOR count) of a matrix $M \in GL(r,\mathbb{F}_{2})$, denoted by d-XOR($M$), is determined by \[\text{d-XOR}(M) = wt(M)-r,\]
	where $wt(M)$ denotes the number of ones in the matrix $M$.
\end{definition}

\begin{definition}~\cite{Lukas_2019}
    The sequential XOR count (s-XOR count) of a matrix $M \in GL(r,\mathbb{F}_{2})$, denoted by s-XOR($M$), is equal to the smallest non-negative integer $t$ such that $M$ can be represented as
    \begin{equation*}
        M=P\prod _{k=1}^{t}{(I+E_{i_k,j_k})},
    \end{equation*}
    where $P$ is a permutation matrix and $E_{i_k,j_k}$, with $i_k\neq j_k$ for all $k$, is a binary matrix with 1 as $(i_k,j_k)$-th entry and 0 elsewhere.
\end{definition}

\noindent When a basis of $\mathbb{F}_{2^r}$ is given, the multiplication of $\alpha \in \mathbb{F}_{2^r}$ given by $x\mapsto \alpha x$ can be expressed as the multiplication of a matrix in $GL(r,\mathbb{F}_{2})$. The matrix depends not only on $\alpha$, but also on the choice of basis of $\mathbb{F}_{2^r}$ over $\mathbb{F}_{2}$. Let $M_{\alpha,B}$ be the matrix representation of the mapping $x\mapsto \alpha x$ with respect to the basis $B$.


\begin{definition}\cite{Lukas_2019}
	Let $\alpha \in \mathbb{F}_{2^r}$. Then the d-XOR count and s-XOR count of $\alpha$, denoted by d-XOR($\alpha$) and s-XOR($\alpha$) respectively, is given by
	\[\text{d-XOR(}\alpha)=~\min_{B} \text{d-XOR}(M_{\alpha,B})~~\text{and}~~\text{s-XOR(}\alpha)=~\min_{B} \text{s-XOR}(M_{\alpha,B}),\]
	where the minimum is taken over all bases of $\mathbb{F}_{2^r}$ over $\mathbb{F}_{2}$.
\end{definition}

The d-XOR count (s-XOR count) of $M_{\alpha,B}$ generally differs from the d-XOR count (s-XOR count) of $M_{\alpha,B'}$ for different bases $B$ and $B'$.
For more details about the two XOR count metrics, see~\cite{Lukas_2019} and the related references mentioned therein.

We denote the XOR count of $\alpha \in \mathbb{F}_{2^r}$ as XOR($\alpha$), which can either be the d-XOR count or the s-XOR count, unless specified otherwise.

The cost of implementing a diffusion matrix can be determined by adding up the XOR counts of each entry in the matrix. If a row has $k_i$ nonzero elements from the field $\mathbb{F}_{2^r}$, these $k_i$ elements must be combined, which incurs a fixed XOR cost of $(k_i-1)r$. Therefore, if an $n$ order matrix has $k_1, k_2, \ldots, k_n$ nonzero elements in its $n$ rows, the matrix incurs a fixed XOR cost of $\sum_{i=1}^{n}(k_i-1)r$ in $\FF_{2^r}$.

The sum, $\mathcal{K}=\sum_{i=1}^{n}(k_i-1)$, is referred to as the fixed XOR of the matrix in this paper. For an MDS matrix of order $n$, its fixed XOR is $\mathcal{K}=n(n-1)$. The XOR count of an $n$ order matrix $M$, denoted by $XOR(M)$, over the field $\mathbb{F}_{2^r}$ is calculated as $\sum_{i,j=1}^{n} {XOR((M)_{i,j})}+\mathcal{K}\cdot r$, where $XOR((M)_{i,j})$ is the XOR count of the entry $(M)_{i,j}$ in $M$.

\noindent Recently, the search for MDS matrices that can be efficiently implemented using global optimization techniques has received a lot of attention~\cite{Banik_Funabiki_Isobe_2019,Duval_Leurent_2018,Kranz_Leander_Stoffelen_Wiemer_2017,Li_Sun_Li_Wei_Hu_2019,Yang_Zeng_Wang_2021}. An implementation that utilizes global optimization can result in significant cost savings compared to local optimization, as common intermediate values can be computed only once and reused. However, this paper focuses on the d-XOR (s-XOR) metric to calculate the implementation cost of the matrices and presents several new NMDS and recursive NMDS matrices with a lower XOR count than previously reported results.

\vspace*{0.1cm}
\textbf{Other Notations:}~Here are the other notations used in the paper.
\begin{enumerate}
	\item The $(i,j)$-th entry of a matrix $A$ is denoted by $(A)_{i,j}$.
	\item We denote by $|A|$ for the number of nonzero entries in the matrix $A$ and $|A|\leq |B|$  means the number of nonzero elements in $A$ is less than or equal to the number of nonzero elements in $B$.
	\item For two $m\times n$ matrices $A$ and $B$, we symbolize $A\leqq  B$ if $(A)_{i,j}\neq 0$ implies $(B)_{i,j}\neq 0$.
	\item $\mathbb{F}_2[L]$ denotes the set of polynomials of $L$ over $\mathbb{F}_2$.
	\item For simplicity, we use nonzero positions in each row of a binary matrix as a representation of the matrix. For example, $[[1,2,3],[1,3],[2]]$ represents the binary matrix $\begin{bmatrix}
	1 & 1 & 1 \\
	1 & 0 & 1 \\
	0 & 1 & 0
	\end{bmatrix}$.
\end{enumerate}

\noindent An MDS matrix must have all its entries nonzero. Therefore, any $n \times n$ matrix cannot be MDS if the number of nonzero entries is less than $n^2$. Whereas, an NMDS matrix of order $n$ must have at least $n^2-n$ nonzero entries. In the following section, we use this fact to obtain some interesting results on NMDS matrices.

\section{Construction of Recursive NMDS matrices}\label{Section_Rec_NMDS}
The importance of recursive MDS matrices is that they are especially well suited for lightweight implementations: the diffusion layer is constructed by recursively executing the implementation of the sparse matrices, which requires some clock cycles.
The use of recursive MDS matrices using companion matrices has been observed in the PHOTON~\cite{PHOTON} family of hash functions and the block cipher LED~\cite{LED} due to their ability to be constructed with a simple LFSR.
Subsequently, Generalized-Feistel-Structure (GFS)~\cite{Recursive_Diffusion_Layer12}, Diagonal-Serial-Invertible (DSI)~\cite{DSI} and diagonal-like sparse (DLS)~\cite{Rec_MDS_2022} matrices were proposed for constructing recursive MDS matrices.
On the other hand, there has been insufficient study on the construction of NMDS matrices, as well as recursive NMDS matrices. This inspires us to present some new results on NMDS matrices and in the following section, we are considering the DLS matrices for the construction of recursive NMDS matrices.

\subsection{Construction of NMDS matrices from DLS matrices}\label{Section_DLS_NMDS}

\begin{definition}\cite{Rec_MDS_2022}\label{def_DLS_matrix_NMDS}
	Let $\rho=~ [i_1,~i_2,~i_3,~\ldots~,i_n]$ be a permutation such that $i_k\neq k$ for $k=1,2,\ldots,n$,  $D_1$ be a nonsingular diagonal matrix and $D_2$ be a diagonal matrix (may be singular). Then we will call the matrix 
	$$B=PD_1+D_2$$ 
	as the diagonal-like sparse (DLS) matrix, where $P$ is the permutation matrix of order $n$ related to the permutation $\rho$. The matrices denoted as $DLS(\rho;D_1,D_2)$.
\end{definition}

\begin{example}
	An example of a DLS matrix of order $4$ is given by
	\begin{center}
		$DLS(\rho;D_1,D_2)=PD_1+D_2=
		\begin{bmatrix}
		0 & 0 & 0 & 1 \\
		0 & 0 & 1 & 0 \\
		1 & 0 & 0 & 0 \\
		0 & 1 & 0 & 0	
		\end{bmatrix}
		\cdot
		\begin{bmatrix}
		a & 0 & 0 & 0 \\
		0 & b & 0 & 0 \\
		0 & 0 & c & 0 \\
		0 & 0 & 0 & d
		\end{bmatrix}
		+
		\begin{bmatrix}
		e & 0 & 0 & 0 \\
		0 & 0 & 0 & 0 \\
		0 & 0 & f & 0 \\
		0 & 0 & 0 & 0
		\end{bmatrix}
		= 
		\begin{bmatrix}
		e & 0 & 0 & d \\
		0 & 0 & c & 0 \\
		a & 0 & f & 0 \\
		0 & b & 0 & 0	
		\end{bmatrix},
		$
	\end{center}
where $P$ is the permutation matrix related to $\rho=[3,4,2,1]$ and $D_1$=$diag(\splitatcommas{a,b,c,d})$ and $D_2=diag(e,0,f,0)$.
\end{example}
In~\cite{Recursive_nMDS_2021}, the authors have obtained some lightweight recursive NMDS matrices with the lowest fixed XOR value of one (i.e. $\mathcal{K}$=1), and it requires a large number of iterations. Thus, these matrices are not useful for low latency purposes. In this paper, we consider the case of checking whether $B^k$ is NMDS or not for $k\leq n$.
\par In~\cite[Theorem~$3$]{Rec_MDS_2022}, the authors proved that for a DLS matrix $M$ of order $n$, $M^k$ contains at least one zero for $0\leq k < n-1$ and $n\geq 2$. In Theorem~\ref{th_DLS_power_need_nMDS}, we have used the similar proof technique to show that for a DLS matrix $M=DLS(\rho;D_1,D_2)$ of order $n$, $M$ cannot be $k$-NMDS for $0\leq k < n-2$. For this, we need the following lemmas.

\begin{lemma}\cite[Lemma~4]{Rec_MDS_2022}\label{lemma_DP=PD_1}
	For any permutation matrix $P$ related to some permutation $\rho$ and any diagonal matrix $D$, we have $DP=PD_1$ for some diagonal matrix $D_1$. Also, the number of nonzero entries in $D$ and $D_1$ are same.
\end{lemma}

\begin{lemma}\cite[Lemma~5]{Rec_MDS_2022}\label{DLS_lemma_power_of_P+D}
	Let $M=P+D_2$ be an $n\times n$ matrix, where $D_2$ is a diagonal matrix (may be singular) and $P$ is a permutation matrix. Then 
	$$M^r\leqq P^r+P^{r-1}D+P^{r-2}D+\cdots +PD+D_{2}^2$$
	for $r\geq 2$, where $D$ denotes some nonsingular diagonal matrix.
\end{lemma}

From Corollary~\ref{corollary_min_nonzero_nMDS}, we know that any matrix of order $n$ cannot be NMDS if the number of nonzero entries is less than $n^2-n$. We are using this fact for the proof of the following theorem.

\begin{theorem}\label{th_DLS_power_need_nMDS}
	Given a DLS matrix $M=DLS(\rho;D_1,D_2)$ of order $n\geq 2$, $M^k$ is not NMDS for $k<n-2$.
\end{theorem}
\begin{proof}
	
	We have $M=DLS(\rho;D_1,D_2)\leqq P+D_2$, where $P$ is the permutation matrix corresponding to $\rho$. From Lemma~\ref{DLS_lemma_power_of_P+D}, we have 
	\begin{equation}\label{Eqn_calculation_nonzeros_in_last_diagonal_matrix}
	\begin{aligned}
		|M^{k}|&\leq |P^{k}+P^{k-1}D+\cdots +PD+D_{2}^2|\\
		&\leq |P^{k}D|+|P^{k-1}D|+\cdots +|PD|+|D_{2}^2|.
	\end{aligned}
	\end{equation}
Since power of a permutation matrix is again a permutation matrix, we have
\begin{equation}\label{Eqn_calculation_nonzero_kth_power_P+D2}
	|M^{k}|\leq \underbrace{|D|+|D|+\cdots +|D|}_{\mbox{$k$ times}}+|D_{2}^2|\leq kn+n.
\end{equation}
Now for $k < n-2$, we have 
\begin{align*}
	|M^k|&< (n-2)n+n= n^2-n
		\implies  |M^k|< n^2-n.
\end{align*}
Hence, $M^k$ is not NMDS for $k<n-2$.\qed
\end{proof}

\begin{remark}
	In Equation~\ref{Eqn_calculation_nonzero_kth_power_P+D2}, if we assume $k<n-1$, we can see that $|M^k|<n^2$. Hence, given a DLS matrix $M$, $M^k$ is not MDS for $k<n-1$.
\end{remark}

\begin{remark}
	From Theorem~\ref{th_DLS_power_need_nMDS}, we know that for a DLS matrix $M=DLS(\rho;D_1,D_2)$ of order $n\geq 2$, $M^k$ is not an NMDS for $k<n-2$. However, there exist $k$-NMDS DLS matrix for $k=n-2$. For example, consider the DLS matrix $M=DLS(\rho; D_1, D_2)$ of order $4$ with $\rho=[4,1,2,3]$, $D_1=diag(\alpha^2,\alpha^2,\alpha^2,\alpha^2)$ and $D_2=diag(\alpha^2,1,\alpha^2,1)$, where $\alpha$ is a primitive element of $\mathbb{F}_{2^4}$ with $\alpha^4+\alpha+1=0$. It can be checked that the matrix
	\begin{equation*}
		\begin{aligned}
			M &= DLS(\rho; D_1, D_2)\\
			& =\begin{bmatrix}
				\alpha^{2} & \alpha^{2} & 0 & 0 \\
				0 & 1 & \alpha^{2} & 0 \\
				0 & 0 & \alpha^{2} & \alpha^{2} \\
				\alpha^{2} & 0 & 0 & 1
			\end{bmatrix}
		\end{aligned}
	\end{equation*}
	is 2-NMDS.
\end{remark}

\noindent We will now explore how the permutation $\rho$ impacts a DLS matrix $DLS(\rho; D_1, D_2)$ in the construction of recursive NMDS matrices. For this, we will use the following lemma.

\begin{lemma}\label{Lemma_Effect_full_cycle}
	If $\rho$ is not an $n$-cycle for a DLS matrix $M=DLS(\rho;D_1,D_2)$ of order $n\geq 2$, then $M^{n-1}$ and $M^n$ contain at most $n^2-2n$ and $n^2-n-2$ nonzero elements respectively.
\end{lemma}

\begin{proof}
	If $\rho$ is not an $n$-cycle, then it is a product of disjoint cycles in $S_n$. Suppose that $\rho=\rho_{1} \rho_2\ldots \rho_v$, where  $\rho_i$ is a $r_{i}$-cycle in $S_n$ for $i=1,2,\ldots,v$ and $v\in \set{2,3,\ldots,\floor{\frac{n}{2}}}$. 
	But by the definition of the DLS matrix, $\rho$ has no fixed points, we have $2\leq r_{i}\leq n-2$ and $r_1+r_2+\cdots+r_v=n$. 
	\par Now from Equation~\ref{Eqn_calculation_nonzero_kth_power_P+D2}, we have $|M^{n-1}|\leq n^2$ and $|M^n|\leq n^2+n$.
	However, we can eliminate some counting of nonzero elements based on the following conditions:
	\begin{enumerate}[1.]
		\item For the permutation matrix $P$ related to $\rho$, $P^{r_{i}}D$ has $r_i$ nonzero elements in the diagonal. Also, $D$ has $n$ many nonzero elements in the diagonal.
		\item $P^{r_{i}+1}D$ and $PD$ have $r_i$ many nonzero elements in the same positions.
		\item Since $v\leq \floor{\frac{n}{2}}$, at least two multiples of some $r_i$ occurs in the set $\set{2,3,\cdots,n}$. Thus, $P^{r_{i}}D$ and $P^{2r_{i}}D$ have at least $r_i$ nonzero elements in the same diagonal position.
	\end{enumerate}
	Therefore, we have
	\begin{align*}
		&|M^{n-1}| \leq n^2-2\cdot (r_1+r_2+\cdots+r_v )\leq n^2-2n\\
		\text{and}~&|M^n| \leq n^2+n-2\cdot (r_1+r_2+\cdots+r_v)-r_i \leq n^2-n-2.
	\end{align*} 
	Hence, the result.\qed
\end{proof}

Also, if $\rho$ is not an $n$-cycle, then by the Condition 1 of the above proof and Equation~\ref{Eqn_calculation_nonzero_kth_power_P+D2}, we can say $|M^{n-2}|< n^2-n$. Hence, by combining the results of Theorem~\ref{th_DLS_power_need_nMDS} and Lemma~\ref{Lemma_Effect_full_cycle}, we have the result as follows.

\begin{corollary}\label{corollary_min_power_in_PD1+D2_if_P_not_full_cycle}
	For a DLS matrix $M=DLS(\rho;D_1,D_2)$ of order $n\geq 2$, if $\rho$ is not an $n$-cycle, then $M^k$ is not NMDS for $k\leq n$.
\end{corollary}


\noindent To this point, we have ignored the possibility that the diagonal of $D_2$ contains zero entries. We now look at the case in which $D_2$ is singular, i.e., its diagonal contains at least one zero.

\begin{lemma}\label{Lemma_nMDS_D2_singular_in_DLS}
	In a DLS matrix $M=DLS(\rho_1;D_1,D_2)$ of order $n\geq 2$, if $D_2$ is singular, then $M^k$ cannot be NMDS for $k\leq n-2$, even if $\rho$ is $n$-cycle.
\end{lemma}
\begin{proof}
	If $D_2$ is singular, having at least one zero in the diagonal then from Equation~\ref{Eqn_calculation_nonzeros_in_last_diagonal_matrix} we have,
	\begin{equation}\label{Equation_singularity_D2}
		\begin{aligned}
			|M^{k}|&\leq |P^{k}D|+|P^{k-1}D|+\cdots +|PD|+|D_{i=0}|\\
	&\leq \underbrace{n+n+\cdots +n}_{\mbox{$k$ times}}+(n-1)=kn+n-1.
		\end{aligned}
	\end{equation}
	Where $D_{i=0}$ be some diagonal matrix with a zero at the $i$-th diagonal position for some $i\in \set{1,2,\ldots,n}$. Thus, for $k\leq n-2$, we have $|M^{k}|< n^2-n.$ Hence, $M$ cannot be $k$-NMDS for $k\leq n-2$.\qed
\end{proof}

\begin{remark}
	When $D_2$ is singular, a DLS matrix $M=DLS(\rho;D_1,D_2)$ of order $n\geq 2$, can be a $k$-NMDS for $k=n-1$. For example, consider the DLS matrix $M=DLS(\rho; D_1, D_2)$ of order $4$ with $\rho=[4,1,2,3]$, $D_1=diag(1,1,1,1)$ and $D_2=diag(1,0,1,0)$ over $\mathbb{F}_{2^r}$.  It can be verified that the matrix
	\begin{equation*}
		\begin{aligned}
			M &= DLS(\rho; D_1, D_2)\\
			& =\begin{bmatrix}
				1 & 1 & 0 & 0 \\
				0 & 0 & 1 & 0 \\
				0 & 0 & 1 & 1 \\
				1 & 0 & 0 & 0
			\end{bmatrix}
		\end{aligned}
	\end{equation*}
	is 3-NMDS.
\end{remark}
\subsubsection{Discussion:}
From \cite[Theorem~4]{Rec_MDS_2022}, we know that for an $n$-MDS DLS matrix of order $n$, we must have $\mathcal{K} \geq \ceil{\frac{n}{2}}$. However, the minimum value of $\mathcal{K}$ may be less than $\ceil{\frac{n}{2}}$ for having at least $n^2-n$ nonzero elements when it is raised to power $n-1$. For example, consider the DLS matrix $B=DLS(\rho;D_1,D_2)$ of order $5$ over $\mathbb{F}_{2^4}$, where $\rho=[5,1,2,3,4]$, $D_1=diag(\alpha^4,\alpha^4,1,1,1)$, $D_2=diag(\alpha,0,0,1,0)$ and $\alpha$ is a root of the constructing polynomial $x^4+x+1=0$. Then $B^4$ have $21$ nonzero elements. However, $B^4$ is not an NMDS matrix. 
\par For NMDS matrices, we could not find the minimum value of $\mathcal{K}$ like we have for MDS matrices. 
In the next section, we will provide some theoretical results about NMDS matrices that will help us in determining the minimum value of $\mathcal{K}$ for DLS matrices of order $n$ that are $k$-NMDS with $k=n-1$ and $k=n$.

\begin{remark}\label{Remark_1_XOR_DLS_not_nMDS}
	In \cite[Theorem 2]{OUR_DSI}, authors prove that any matrix $M$ of order $n$ with $\mathcal{K}=1$ can have at most $\frac{n(n+3)}{2}-1$ nonzero elements when it raised to the power $n$. Hence, for $n\geq 5$, we have $|M^n|< n^2-n$. Thus, for $n\geq 5$, any matrix of order $n$ with $\mathcal{K}=1$ cannot be $n$-NMDS.
\end{remark}
For $n=4$, we have $\frac{n(n+3)}{2}-1 > n^2-n$. So it may seem that a matrix of order 4 with $\mathcal{K}=1$ can be 4-NMDS. However, in the following theorem, we will see that to be 4-NMDS, a matrix of order 4 must have $\mathcal{K}=2$.

\begin{theorem}\label{Th_4-nMDS_with_1_XOR}
	There does not exist any $4$-NMDS matrix of order $4$ with $\mathcal{K}=1$ over a field of characteristic $2$.
\end{theorem}

\begin{proof}
	A matrix $M$ of order $n$ can never be recursive NMDS if its one row or column has all zero entries~\footnote{The theorem states about the matrices of order $n=4$ and the first part of the proof holds for any matrix of order $n$.}. Also, if $M$ contains $n$ many nonzero elements in such a way that no column or row has all zero entries, then $M$ is of the form $M=PD$, where $P$ is a permutation matrix and $D$ is a diagonal matrix. Then by Lemma~\ref{lemma_DP=PD_1}, any power of $M$ is again of the form $P'D'$, for some permutation matrix $P'$ and diagonal matrix $D'$. Hence, $M$ cannot be recursive NMDS.

	\par Let $\mathcal{S}$ be the set of all matrices $M$ that contain $n+1$ many nonzero elements with $\mathcal{K}=1$ and in such a way that no column or row has all zero entries. Then each $M\in \mathcal{S}$ can be written as $M=PD+A$, where $A$ has only one nonzero element. Let the nonzero element lies in the $i$-th row of $A$.

	\par Now consider the permutation matrix $P_1$ obtained from the identity matrix by permuting the row $i$ to row $1$. Now
	\begin{align*}
		P_1MP_1^{-1}&= P_1(PD+A)P_1^{-1}\\
		&=P_1PDP_1^{-1}+P_1AP_1^{-1}
	\end{align*}
	By Lemma~\ref{lemma_DP=PD_1}, we have $DP_1^{-1}=P_1^{-1}D'$ for some diagonal matrix $D'$. Thus we have
	\begin{align*}
		P_1MP_1^{-1}&=P_1PP_1^{-1}D'+P_1AP_1^{-1}\\
		&=QD'+A',
	\end{align*}
	where $Q= P_1PP_1^{-1}$, $A'=P_1AP_1^{-1}$ and $A'$ has the nonzero element in its first row. Therefore, $M$ is permutation similar to $QD'+A'$. Now let $\mathcal{S}' \subset \mathcal{S}$ be the set of all matrices with two nonzero elements in the first row. 
	\par Since from Fact~\ref{Fact_perm_similar_mds_nmds}, we know that a permutation similar to a recursive NMDS matrix is also a recursive NMDS matrix, we simply need to check from the set $\mathcal{S}'$ for finding all recursive NMDS matrices with $\mathcal{K}=1$.
	
	\par It can be checked that there are only six\footnote{For $n=4$, there are a total of $72$ elements in $\mathcal{S}'$, and by running a computer search, we have observed that there are only $6$ matrix structures that have at least 12 nonzero elements when raised to the power $4$.} matrix structures (See~(\ref{Eqn_6_matrix_with_4nMDS})) of order $n=4$ from the set $\mathcal{S}'$ that can potentially be NMDS (i.e. number of nonzero elements >$12$) when they are raised to power $4$. However, all the six structures 

	\begin{equation}\label{Eqn_6_matrix_with_4nMDS}
		\begin{bmatrix}
			* & * & 0 & 0\\
			0 & 0 & * & 0\\
			0 & 0 & 0 & *\\
			* & 0 & 0 & 0
		\end{bmatrix},
		\begin{bmatrix}
			* & 0 & 0 & *\\
			* & 0 & 0 & 0\\
			0 & * & 0 & 0\\
			0 & 0 & * & 0
		\end{bmatrix},
		\begin{bmatrix}
			* & * & 0 & 0\\
			0 & 0 & 0 & *\\
			* & 0 & 0 & 0\\
			0 & 0 & * & 0
		\end{bmatrix},
		\begin{bmatrix}
			* & 0 & * & 0\\
			* & 0 & 0 & 0\\
			0 & 0 & 0 & *\\
			0 & * & 0 & 0
		\end{bmatrix},
		\begin{bmatrix}
			* & 0 & 0 & *\\
			0 & 0 & * & 0\\
			* & 0 & 0 & 0\\
			0 & * & 0 & 0
		\end{bmatrix}~\text{and}~
		\begin{bmatrix}
			* & 0 & * & 0\\
			0 & 0 & 0 & *\\
			0 & * & 0 & 0\\
			* & 0 & 0 & 0
		\end{bmatrix}
	\end{equation}
	are also permutation similar.
	Now consider the first matrix structure and let 
	\begin{equation*}
		M=\begin{bmatrix}
			a & x_1 & 0 & 0\\
			0 & 0 & x_2 & 0\\
			0 & 0 & 0 & x_3\\
			x_4 & 0 & 0 & 0
		\end{bmatrix},
	\end{equation*}
	where $a,x_1,x_2,x_3$ and $x_4$ are some nonzero elements in the field. Now consider the input vector of $M$ as $[1,ax_1^{-1},0,0]^{T}$. The resultant vector after each iteration is 
	\begin{align*}
		&\begin{bmatrix}
			1\\
			ax_1^{-1}\\
			0\\
			0
		\end{bmatrix}
		\xrightarrow[\text{i=1}]{\text{}}
		\begin{bmatrix}
			0\\
			0\\
			0\\
			x_4
		\end{bmatrix}
		\xrightarrow[\text{i=2}]{\text{}}
		\begin{bmatrix}
			0\\
			0\\
			x_3x_4\\
			0
		\end{bmatrix}
		\xrightarrow[\text{i=3}]{\text{}}
		\begin{bmatrix}
			0\\
			x_2x_3x_4\\
			0\\
			0
		\end{bmatrix}
		\xrightarrow[\text{i=4}]{\text{}}
		\begin{bmatrix}
			x_1x_2x_3x_4\\
			0\\
			0\\
			0
		\end{bmatrix}
	\end{align*}
	The sum of nonzero elements of input vector and output vector in each iteration is less than 4 i.e. branch number of $M<4$. Therefore, $M$ is not $k$-NMDS for $k\leq 4$. Hence, there does not exist any $4$-NMDS matrix of order $4$ with $\mathcal{K}=1$ over a field of characteristic 2.\qed
\end{proof}


%
From~\cite[Lemma 3]{OUR_DSI}, we can easily check that for a matrix of order $n\geq 4$ with $\mathcal{K}=1$, $|M^{k}|< n^2-n$ for $k\leq n-1$. Thus, by using Remark~\ref{Remark_1_XOR_DLS_not_nMDS} and Theorem~\ref{Th_4-nMDS_with_1_XOR}, we have the following result.

\begin{theorem}\label{Th_1_XOR_matrix_not_NMDS}
	For $n\geq 4$, there does not exist any $k$-NMDS matrix of order $n$ with $\mathcal{K}=1$ and $k\leq n$ over a field of characteristic $2$.
\end{theorem}

\begin{remark}
	For $n<4$, there may exist a $k$-NMDS matrix of order $n$ with $\mathcal{K}=1$ and $k\leq n$. For example, the matrix $B=
	\begin{bmatrix}
		0 & 1 & 0 \\
		0 & 0 & 1 \\
		1 & 1 & 0
	\end{bmatrix}$ is $3$-NMDS.
\end{remark}

\begin{fact}
	Over a field of characteristic $2$, a DLS matrix of order $n$ with $\mathcal{K}=1$ cannot be $k$-NMDS for $n\geq 4$ and $k\leq n$.
\end{fact}
Now we will discuss some equivalence classes of DLS matrices for the construction of recursive NMDS matrices.

\subsection{Equivalence classes of DLS matrices} 
If the DLS matrix $DLS(\rho;D_1,D_2)$ of order $n$ has fixed XOR $\mathcal{K}=l$, the diagonal of $D_2$ has $l$ nonzero elements. Therefore, there are $\Comb{n}{l}$ possible arrangements for distributing these $l$ nonzero elements along the diagonal of $D_2$.

Also, in a DLS matrix $DLS(\rho;D_1,D_2)$, the permutation $\rho=[i_1,i_2,\ldots,i_n]$ must satisfy $i_k\neq k$ for $k=1,2,\ldots,n$. In other words, $\rho$ represents a derangement of a set of $n$ elements. Therefore, there are $D(n)$\footnote{The formula for $D_n$ is given by $D_n=(n-1)[D_{n-1}+D_{n-2}]$ with initial conditions $D_1=1$ and $D_0=0$. For example, the values of $D(n)$ are $1,2,9,44$, and $265$ for $n=2,3,4,5$, and $6$, respectively.} possible choices for $\rho$ in a DLS matrix, where $D(n)$ denotes the number of derangements of a set of $n$ elements~\cite{combinatorics}.

As a result, the search space for finding a recursive NMDS matrix from the DLS matrices over the field $\mathbb{F}_{2^r}$ is  given by $D(n)\cdot \Comb{n}{l}\cdot (2^r)^{(n+l)}$. For example, the search space for finding a $6$-NMDS matrix from a DLS matrix of order $6$, with $\mathcal{K}=3$, over the field $\mathbb{F}_{2^4}$ is $265\cdot 20\cdot 2^{36}\approx 2^{48}$.

However, we have drastically reduced the search space by defining some equivalence classes of DLS matrices. 

\begin{theorem}\label{th_DLS_MDS_iff_P+D2_NMDS}
	Let $a,a_1,a_2,\ldots,a_n,a_1',a_2',\ldots,a_n'\in \FF_{2^r}^*$ such that $a=\prod_{i=1}^{n}{a_i}=\prod_{i=1}^{n}{a_i'}$. Then for any diagonal matrix $D_2$ over $\mathbb{F}_{2^r}$, the DLS matrix $M=DLS(\rho;D_1,D_2)$ of order $n$ is $k$-NMDS if and only if $M'=DLS(\rho;D_1',D_2)$ is $k$-NMDS, where $k\in \set{n-1,n}$, $D_1=diag(a_1,a_2,\ldots,a_n)$ and $D_1'=diag(\splitatcommas{a_1',a_2',\ldots,a_n'})$.
\end{theorem}

\begin{proof}
	Suppose $\rho=~ [\splitatcommas{i_1,~i_2,~i_3,~\ldots~,i_n}]$ and $P$ is the permutation matrix corresponding to $\rho$.
	
	Now for any nonsingular diagonal matrix $D=diag(d_1,d_2,\ldots,d_n)$, we have	
	\begin{align*}
	DMD^{-1} &= D(PD_1+D_2)D^{-1} = DPD_1D^{-1}+D_2.
	\end{align*}
	Now by Lemma~\ref{lemma_DP=PD_1}, we have $DP=PD'$ where $D'=diag(d_{i_1},d_{i_2},\ldots,d_{i_n})$. Thus, we have	
	$$DMD^{-1} = P(D'D_{1}D^{-1})+D_2.$$
	If $D'D_{1}D^{-1}=D_1'$, then we have 
	\begin{equation}\label{Eqn_D(aM)D^{-1}=P+D_2}
	\begin{aligned}
	DMD^{-1} &= PD_1'+D_2=M'.
	\end{aligned}
	\end{equation}	
	Now if $D'D_{1}D^{-1}=D_1'$, we have $D'D_{1}=D_1'D$. Therefore, we have	
	\begin{equation}\label{Eqn_a=a1a2...an}
	\left.
	\begin{array}{ll}
	a_1d_{i_1} & = a_1'd_1\\
	a_2d_{i_2} & = a_2'd_2\\
	&\vdots\\
	a_nd_{i_n} & = a_n'd_n
	\end{array}
	\right \}
	\quad ~\implies ~\quad 
	\left \{\
	\begin{array}{ll}
	d_1 & =a_1d_{i_1}(a_1')^{-1}\\
	d_2 & =a_2d_{i_2}(a_2')^{-1}\\
	&\vdots\\
	d_n & =a_nd_{i_n}(a_n')^{-1}.
	\end{array}
	\right.
	\end{equation}
	From Corollary~\ref{corollary_min_power_in_PD1+D2_if_P_not_full_cycle}, we know that a DLS matrix of order $n$ can be $k$-NMDS, with $k\leq n$, only when $\rho$ is $n$-cycle. Thus, since $\rho$ is $n$-cycle, we have 
	$$a_1a_2\ldots a_n=a_1'a_2'\ldots a_n'=a.$$
	Now choosing $d_1=1$, from Equation~\ref{Eqn_a=a1a2...an}, we get the values of other $d_j$'s in terms of $a_i$'s and $a_i'$'s for $j=2,3,\ldots,n$.
	Also, from Corollary~\ref{corollary_DAD^(-1)_MDS} and Equation~\ref{Eqn_D(aM)D^{-1}=P+D_2}, we can say that $M$ is $k$-NMDS if and only if $M^{'}$ is $k$-NMDS.\qed
\end{proof}

\begin{corollary}~\label{corollary_equivalence_DLS_NMDS}
	Let $a=\prod_{i=1}^{n}{a_i}$ for $a_1,a_2,\ldots,a_n \in \FF_{2^r}^*$. Then for any diagonal matrix $D_2$ over $\mathbb{F}_{2^r}$, the DLS matrix $M=DLS(\rho;D_1,D_2)$ of order $n$ is $k$-NMDS if and only if $M^{'}=DLS(\rho;D_1^{'},D_2)$ is $k$-NMDS, where $k\in \set{n-1,n}$, $D_1=diag(a_1,a_2,\ldots,a_n)$ and $D_1^{'}=diag(a,1,1,\ldots,1)$.
\end{corollary}
\begin{remark}
	For any $c\in \FF_{2^r}^*$, $M$ is $k$-NMDS implies $cM$ is also $k$-NMDS. Thus, if $\rho$ is an $n$-cycle permutation, $M= DLS(\rho;D_1,D_2)$ is diagonal similar to $M'=DLS(\rho;D_1',D_2')$, where $D_1=diag(a_1,a_2,\ldots,a_n)$, $D_1'=diag(c^n a,1,1,\ldots,1)$, $D_2'=c\cdot D_2$ and $a=\prod_{i=1}^{n}{a_i}$. We know that $x\mapsto x^{2^l}$ is an isomorphism over $\FF_{2^r}$. So when $n=2^l$, there exist an element $c=a^{-1/n} \in \FF_{2^r}^*$. Hence, when $n=2^l$, we can say that $M$ is diagonal similar to $M''=DLS(\rho;D_1'', D_2'')$, where $D_1''=diag(1,1,1,\ldots,1)$ and $D_2''$ is some diagonal matrix. Therefore, for $k\in \set{n-1,n}$, $M$ is $k$-NMDS if and only if $M''$ is also $k$-NMDS.
\end{remark}
We know that a permutation similar to an NMDS matrix is again an NMDS matrix. Thus, we can reduce the search space further by eliminating the permutation similar matrices from the search space. For this, we need the following lemma.

\begin{lemma}~\label{Lemma_single_n_cycle_for_DLS}
	Let $M_1=DLS(\rho_1;D_1,D_2)$ be a DLS matrix of order $n$ and $\rho_2 \in S_n$ is conjugate with $\rho_1$, then $M_1$ is $k$-NMDS if and only if $M_2=DLS(\rho_2;D_1^{'},D_2^{'})$ is $k$-NMDS, where $D_1^{'}$ and $D_2^{'}$ are some diagonal matrices.
\end{lemma}

\begin{proof}
	Since $\rho_1$ and $\rho_2$ are conjugate, we have $\sigma \rho_1 \sigma^{-1}=\rho_2$,
	for some $\sigma \in S_n$. Let $P_1,P_2$ and $P$ be the permutation matrices related to $\rho_1,\rho_2$ and $\sigma$ respectively. Then we have 
	\begin{align*}
		PM_1P^{-1}&= P(P_1D_1+D_2)P^{-1} = PP_1D_1P^{-1}+PD_2P^{-1}\\
				  &= PP_1P^{-1}D_1^{'}+PP^{-1}D_2^{'}, 
	\end{align*}
	where $D_1P^{-1}=P^{-1}D_1^{'}$ and $D_2P^{-1}=P^{-1}D_2^{'}$ for some diagonal matrices $D_1^{'}$ and $D_2^{'}$. Thus, we have
	$PM_1P^{-1}= P_2D_1^{'}+D_2^{'}= M_2.$
	Since $PM_1P^{-1}=M_2$, from Corollary~\ref{corollary_PAP^(-1)_MDS} we can say that $M_1$ is $k$-NMDS if and only if $M_2$ is $k$-NMDS.\qed
\end{proof}

\begin{remark}~\label{Remark_single_n_cycle_for_DLS}
	If $D_2$ is singular, a DLS matrix $M=DLS(\rho_1;D_1,D_2)$ of order $n$ cannot be $k$-NMDS for $k\leq n-2$. Also, $\rho$ must be an $n$-cycle for $M$ to be $k$-NMDS with $k=n-1$ or $k=n$. In addition, the $n$-cycles in $S_n$ are conjugate with each other. Therefore, to find the $k$-NMDS (with $k = n -1$ and $k = n$) DLS matrices, we need to check only for the DLS matrices associated with one fixed $n$-cycle $\rho$.
\end{remark}
Now consider $\mathbb{D}(n,\mathbb{F}_{2^r})$ to be the set of all DLS matrices $DLS(\rho;D_1,D_2)$ of order $n$, with $\mathcal{K}=\ceil {\frac{n}{2}}$, over the field $\mathbb{F}_{2^r}$ and define
\begin{align*}
	\mathbb{D}^{'}(n,\mathbb{F}_{2^r})&= \{ B\in \mathbb{D}(n,\mathbb{F}_{2^r}):~B=P^{'}D_1^{'}+D_2^{'} \}, 
\end{align*}
where $P^{'}$ is the permutation matrix related to the $n$ length cycle~$[2,3,4,\ldots, n-1,n,1]$\footnote{By Remark~\ref{Remark_single_n_cycle_for_DLS}, any $n$ length cycle can be chosen for the set $\mathbb{D}^{'}(n,\mathbb{F}_{2^r})$.} and $D_1^{'}=diag(a,1,1,\ldots,1)$. 

Thus, to find the $k$-NMDS (with $k = n -1$ and $k = n$) DLS matrices over $\mathbb{F}_{2^r}$, we need to check only for the DLS matrices in the set $\mathbb{D}^{'}(n,\mathbb{F}_{2^r})$.

\begin{remark}
	From the discussion of~\cite[Section~3.3]{Rec_MDS_2022}, we know that if $\rho=[\splitatcommas{2,3,4,\ldots, n-1,n,1}]$ and $D_2$ has any two consecutive zero entries, then the DLS matrix of order $n$ cannot be $n$-MDS. However, this result is not true for NMDS matrices. For example, consider the DLS matrix $B=DLS(\rho;D_1,D_2)$ of order $4$ over $\FF_{2^4}$, where $\rho=[2, 3, 4, 1]$, $D_1=diag(1,1,1,1)$, $D_2=diag(1,\alpha,0,0)$ and $\alpha$ is a primitive element of $\FF_{2^4}$ with $\alpha^4+\alpha+1=0$. Then it can be checked that $B$ is $4$-NMDS.
\end{remark}
Thus, for finding $k$-NMDS DLS matrices with $k\in \set{n-1,n}$ and $\mathcal{K}=l$, we need to check for all the $\Comb{n}{l}$ arrangements for the $l$ nonzero elements in the diagonal of $D_2$. Thus, the search space for finding $k$-NMDS DLS matrices, with $k\in \set{n-1,n}$, over the field $\mathbb{F}_{2^r}$ has been reduced from $D(n)\cdot \Comb{n}{l}\cdot (2^r)^{(n+l)}$ to $\Comb{n}{l}\cdot (2^r)^{(1+l)}$. Then, by exhaustive search in the restricted domain, we have the results for the existence of $k$-NMDS DLS matrices over $\FF_{2^4}$ and $\FF_{2^8}$ for $n= 4, 5, 6, 7, 8$ listed in Table~\ref{table_search_all_k-nMDS}.


\begin{table}
	\begin{center}
		\begin{minipage}{\textwidth}
		\caption{$k$-NMDS DLS matrix of order $n$ over the field $\mathbb{F}_{2^r}$ with $k=n-1$ and $k=n$~(``\textbf{DNE}'' stands for does not exist).	\label{table_search_all_k-nMDS}}
		\begin{tabular}{ |c|c||c|c|c|c|c|c| }
			\hline
			& & \multicolumn{2}{c|}{$\mathcal{K}=2$} & \multicolumn{2}{c|}{$\mathcal{K}=3$} & \multicolumn{2}{c|}{$\mathcal{K}=4$} \\ \cline{3-8}
			Order $n$ & ~~$k$~~ & over $\mathbb{F}_{2^4}$ & over  $\mathbb{F}_{2^8}$ & over $\mathbb{F}_{2^4}$ & over $\mathbb{F}_{2^8}$ & over $\mathbb{F}_{2^4}$ & over $\mathbb{F}_{2^8}$ \\ \hline
			4 & 3 & Exists & Exists & -- & -- & -- & -- \\ 
			  & 4 & Exists & Exists & -- & -- & -- & -- \\ \hline
			  5 & 4 & DNE & DNE & Exists & Exists & -- & -- \\ 
			  & 5 & DNE & DNE & Exists & Exists & -- & -- \\ \hline
			  6 & 5 & DNE & DNE & Exists & Exists & -- & -- \\ 
			  & 6 & DNE & DNE & Exists & Exists & -- & -- \\ \hline
			  7 & 6 & DNE & DNE & DNE & * & Exists & Exists \\ 
			  & 7 & DNE & DNE & DNE & * & Exists & Exists \\ \hline
			  8 & 7 & DNE & DNE & DNE & DNE & DNE & Exists \\ 
			  & 8 & DNE & DNE & DNE & DNE & DNE & Exists \\ \hline
		\end{tabular}
		\vskip3pt
		$^*$ Over $\mathbb{F}_{2^8}$, we are unable to make a decision for $n = 7$ with $\mathcal{K} = 3$ since we were unable to perform an exhaustive search even in the restricted domain.
	\end{minipage}
	\end{center}
\end{table}

\section{Construction of Recursive NMDS matrices from GDLS matrices}\label{Section_Rec_NMDS_GDLS}
In this section, we present some lightweight recursive NMDS matrices of orders $4, 5, 6, 7$, and $8$ from the  GDLS matrices, introduced in~\cite{Rec_MDS_2022}. 

\begin{definition}\cite{Rec_MDS_2022}\label{def_GDLS_matrix_NMDS}
	Let $\rho_1=~ [i_1,~i_2,~i_3,~\ldots~,i_n]$ and $\rho_2=~ [j_1,~j_2,~j_3,~\ldots~,j_n]$ be two permutations such that $i_k\neq j_k$ for $k=1,2,\ldots,n$,  $D_1$ be a nonsingular diagonal matrix and $D_2$ be a diagonal matrix (may be singular). Then we will call the matrix 
		$$B=P_1D_1+P_2D_2$$
	as the generalized DLS (GDLS) matrix, where $P_1$ and $P_2$ are the permutation matrices of order $n$ related to the permutation $\rho_1$ and $\rho_2$ respectively. We will denote these matrices as $GDLS(\rho_1,\rho_2;D_1,D_2)$.
\end{definition}
From Lemma~\ref{Lemma_nMDS_D2_singular_in_DLS}, we know that if $D_2$ is singular, a DLS matrix $DLS(\rho;D_1,D_2)$ can never be $k$-NMDS for $k\leq n-2$. However, this result is not applicable to GDLS matrices. For instance, the GDLS matrix $M=GDLS(\rho_1,\rho_2;D_1,D_2)$ of order $7$ with $\rho_1=[6,7,4,5,2,3,1],~\rho_2=[3,2,1,4,7,6,5],~D_1=diag(\splitatcommas{1, 1, 1, 1, 1, 1, \alpha})$ and $D_2=diag(1, 0, \alpha^2, 0, \alpha, 0, \alpha^2)$ is 5-NMDS, where $\alpha$ is a primitive element of $\mathbb{F}_{2^4}$ with $\alpha^4+\alpha+1=0$.

Since GDLS matrices have the potential to generate NMDS matrices with fewer iterations, we select them for constructing recursive NMDS matrices. To find recursive NMDS matrix, we begin with $k=n-2$, and if this does not yield a result, we increase the value of $k$.

From the definition of GDLS matrices, it can be observed that the size of the set of all GDLS matrices with $\mathcal{K}=l$ over the field $\mathbb{F}_{2^r}$ is $n!\cdot D(n)\cdot \Comb{n}{l}\cdot (2^r)^{(n+l)}$, where $D(n)$ represents the number of derangements for $n$ distinct objects. This size is extremely large, making an exhaustive search impractical for obtaining a $k$-NMDS matrix of order $n\geq 5$ from the GDLS matrices.

To minimize the search space, in most cases, we arbitrarily select $\rho_1$ as the $n$-cycle $[\splitatcommas{n, 1, 2, \ldots, n-1}]$. However, it is important to note that there is no inherent advantage in choosing $\rho_1=[n, 1, 2, \ldots, n-1]$ for obtaining a recursive NMDS matrix. If we change $\rho_1=[n, 1, 2, \ldots, n-1]$ to any permutation from $S_n$, there is still a possibility of obtaining a recursive NMDS matrix.

Also, to find lightweight recursive NMDS matrices, we looked through the GDLS matrices of order $n$ with $\mathcal{K}=\ceil{\frac{n}{2}}$ whose entries are from the set $\{\splitatcommas{1,\alpha,\alpha^{-1},\alpha^{2},\alpha^{-2}}\}$, where $\alpha$ is a primitive element and a root of the constructing polynomial of the field $\FF_{2^r}$. The search space for finding $k$-NMDS matrices of order $n \geq 5$ remains large, even when considering the set $\{\splitatcommas{1,\alpha,\alpha^{-1},\alpha^{2},\alpha^{-2}}\}$. Therefore, to obtain $k$-NMDS matrices of order $n=5,6,7,8$, we conduct a random search. 

Also note that the implementation costs of the matrices presented in this section over a field are calculated by referring to the s-XOR count value of the corresponding field elements as provided in table of~\cite[App. B]{DSI}.

\subsection{Construction of $4\times 4$ Recursive NMDS matrices}\label{Section_4-near-MDS_GDLS}
In this section, we propose a GDLS matrix $B$ of order $4$ that yields a recursive NMDS matrix over the field $\mathbb{F}_{2^r}$ for $r\geq 1$. Based on Theorem~\ref{Th_1_XOR_matrix_not_NMDS}, it is known that there are no $k$-NMDS matrices of order $4$ with $\mathcal{K}=1$ and $k\leq 4$ over a field of characteristic 2. Therefore, to obtain recursive NMDS matrices of order 4, we must choose $\mathcal{K}\geq 2$.
The proposed GDLS matrix is constructed by the permutations $\rho_1= [2, 3, 4, 1],~\rho_2= [1, 2, 3, 4]$ and diagonal matrices $D_1=diag(1,1,1,1)$, $D_2=diag(0,1,0,1)$.

\begin{equation}\label{Eqn_4-nMDS_over_field_F_2^r}
	\begin{aligned}
		B&= GDLS(\rho_1,\rho_2;D_1,D_2)=
		\begin{bmatrix}
			0 & 0 & 0 & 1 \\
			1 & 1 & 0 & 0 \\
			0 & 1 & 0 & 0 \\
			0 & 0 & 1 & 1
		\end{bmatrix}
	\end{aligned}
\end{equation}
The matrix $B$ is a $3$-NMDS matrix with a XOR count of $2\cdot r=2r$ over the field $\mathbb{F}_{2^r}$.

\begin{lemma}\label{Lemma_lowest_XOR_4nMDS_8r}
	For $k$-NMDS matrix of orders $4$ with $k\leq 4$, the lowest XOR count is $2r$ over the field $\mathbb{F}_{2^r}$.
\end{lemma}

\begin{proof}
	From Remark~\ref{Remark_1_XOR_DLS_not_nMDS} and Theorem~\ref{Th_4-nMDS_with_1_XOR}, we know that any matrix $M$ of order $4$ with $\mathcal{K}=1$ cannot be $k$-NMDS for $k\leq 4$. Hence, we must have $\mathcal{K} \geq 2$. Therefore, we have $XOR(M)\geq 2\cdot r$ over the field $\mathbb{F}_{2^r}$. \qed
\end{proof}
\begin{remark}\label{Remark_lowest_XOR_4_rec_NMDS}
	For $k\leq 4$, the proposed matrix $B$ in (\ref{Eqn_4-nMDS_over_field_F_2^r}) has the lowest XOR count among the $k$-NMDS matrices of order 4 over the field $\mathbb{F}_{2^r}$ for $r\geq 1$.
\end{remark}
\subsection{Construction of $5\times 5$ Recursive NMDS matrices}\label{Section_5-near-MDS_GDLS}
This section presents two GDLS matrices, $A_1$ and $A_2$, of order $5$ that give NMDS matrices when raised to power $4$ and $5$, respectively, over the field $\mathbb{F}_{2^4}$.
We also looked for GDLS matrices $M$ of order $5$ such that $M^k$ is NMDS for $k\leq n-2$ and $\mathcal{K}=3$, but we were unable to find any over $\FF_{2^4}$.
Consider the GDLS matrices $A_1$ and $A_2$ of order $5$ which are constructed as follows:
\begin{enumerate}[(i)]
	\item $A_1$:~$\rho_1=[5, 1, 2, 3, 4], \rho_2=[3, 2, 5, 4, 1]$, $D_1=diag(1,1,1,1,1)$ and $D_2=diag(0, \alpha, 0, 1, \alpha^{-1})$
	\item $A_2$:~ $\rho_1=[5, 1, 2, 3, 4], \rho_2=[3, 4, 5, 1, 2]$, $D_1=diag(1,1,1,1,1)$ and $D_2=diag(0, 1, 0, 1, \alpha)$
\end{enumerate}
\begin{equation}\label{Equ_5-nMDS_over_field_F_2^4}
	\begin{aligned}
		A_1&=
		\begin{bmatrix}
			0 & 1 & 0 & 0 & \alpha^{-1} \\
			0 & \alpha & 1 & 0 & 0 \\
			0 & 0 & 0 & 1 & 0 \\
			0 & 0 & 0 & 1 & 1 \\
			1 & 0 & 0 & 0 & 0
		\end{bmatrix}~~
		A_2=
		\begin{bmatrix}
			0 & 1 & 0 & 1 & 0 \\
			0 & 0 & 1 & 0 & \alpha \\
			0 & 0 & 0 & 1 & 0 \\
			0 & 1 & 0 & 0 & 1 \\
			1 & 0 & 0 & 0 & 0
		\end{bmatrix},
	\end{aligned}
\end{equation}
where $\alpha$ is a primitive element of~$\FF_{2^4}$ and a root of $x^4+x+1$. It is easy to verify that the matrix $A_1$ is a $4$-NMDS matrix with a XOR count of $(1+1)+3\cdot 4=14$ and $A_2$ is a $5$-NMDS matrix with a XOR count of $1+3\cdot 4=13$.

\noindent In Lemma~\ref{Lemma_lowest_XOR_count_5-nMDS}, we discuss the lowest XOR count of recursive NMDS matrices of order $n\geq 5$. For this, we need the following result from~\cite{Choy_2008}.

\begin{theorem}\cite{Choy_2008}\label{Th_max_branch_number_binary_matrix}
	A matrix of order $n$, with 0 and 1 as entries, has a maximum branch number of $\frac{2n+4}{3}$.
\end{theorem}

\begin{lemma}\label{Lemma_lowest_XOR_count_5-nMDS}
	Given a recursive NMDS matrix $B$ of orders $n\geq 5$, with $\mathcal{K}=l$, the lowest XOR count of $B$ is $XOR(\beta)+l\cdot r$ over the field $\mathbb{F}_{2^r}$, where $\beta$ ($\neq 1$) is a nonzero element in $\mathbb{F}_{2^r}$ with the lowest XOR count value in that field.
\end{lemma}

\begin{proof}
	An NMDS matrix of order $n$ has branch number of $n$. So from Theorem~\ref{Th_max_branch_number_binary_matrix}, we can say that a matrix of order $n$, with entries from the set $\set{0,1}\subseteq \mathbb{F}_{2^r}$, cannot be NMDS for $n\geq 5$. If we take a matrix $B$ with entries of 0 or 1, then the entries of $B^k$ will remain in the set $\set{0,1}$ for any power $k$. So $B$ must have an element $\gamma \not \in \set{0,1}$. Therefore, $XOR(B)\geq $ $XOR(\beta)+l\cdot r$, where $\beta$ ($\neq 1$) is a nonzero element in $\mathbb{F}_{2^r}$ with the lowest XOR count value in that field.\qed
\end{proof}

\begin{remark}
	Over the field $\mathbb{F}_{2^4}$, the matrix $A_2$ in (\ref{Equ_5-nMDS_over_field_F_2^4}) has the lowest XOR count among the $5$-NMDS matrices of order $5$ and $\mathcal{K}=3$.
\end{remark}

\subsection{Construction of $6\times 6$ Recursive NMDS matrices}\label{Section_6-near-MDS_GDLS}
In this section, we introduce two lightweight GDLS matrices, $B_1$ and $B_2$, of order $6$ with $\mathcal{K}=3$. These matrices can be implemented with $14$ and $13$ XORs over the field $\mathbb{F}_{2^4}$, respectively, and yield NMDS matrices when raised to the power of $5$ and $6$, respectively.
The matrices $B_1$ and $B_2$ of order $6$ are constructed as follows:
\begin{enumerate}[(i)]
	\item $B_1$:~$\rho_1=[6, 1, 2, 3, 4, 5], \rho_2=[1, 2, 3, 4, 5, 6]$, $D_1=diag(1,1, 1,1, 1,1)$ and $D_2=diag(0,\alpha, 0,\alpha^{-1}, 0,1)$
	\item $B_2$:~ $\rho_1=[6, 1, 2, 3, 4, 5], \rho_2=[3, 4, 5, 2, 6, 1]$, $D_1=diag(1,1, 1,1, 1,1)$ and $D_2=diag(0,\alpha, 0,1, 0,1)$
\end{enumerate}

\begin{equation}\label{Equ_6-nMDS_over_field_F_2^4}
	\begin{aligned}
		B_1&=
		\begin{bmatrix}
			0 & 1 & 0 & 0 & 0 & 0 \\
			0 & \alpha & 1 & 0 & 0 & 0 \\
			0 & 0 & 0 & 1 & 0 & 0 \\
			0 & 0 & 0 & \alpha^{-1} & 1 & 0 \\
			0 & 0 & 0 & 0 & 0 & 1 \\
			1 & 0 & 0 & 0 & 0 & 1
		\end{bmatrix}~~
		B_2=
		\begin{bmatrix}
			0 & 1 & 0 & 0 & 0 & 1 \\
			0 & 0 & 1 & 1 & 0 & 0 \\
			0 & 0 & 0 & 1 & 0 & 0 \\
			0 & \alpha & 0 & 0 & 1 & 0 \\
			0 & 0 & 0 & 0 & 0 & 1 \\
			1 & 0 & 0 & 0 & 0 & 0
		\end{bmatrix},
	\end{aligned}
\end{equation}
where $\alpha$ is a primitive element of~$\FF_{2^4}$ and a root of $x^4+x+1$. It can be checked that the matrix $B_1$ is a $5$-NMDS matrix with a XOR count of $(1+1)+3\cdot 4=14$ and $B_2$ is a $6$-NMDS matrix with a XOR count of $1+3\cdot 4=13$.

\noindent We also searched for GDLS matrices $M$ of order $6$ such that $M^k$ is NMDS for $k\leq n-2$ and $\mathcal{K}=3$, but we could not find such matrices over $\mathbb{F}_{2^4}$.

\begin{remark}\label{Example_k-nMDS_order_5_6_ring_F_2^8}
	It is not possible to have elements with XOR count $1$ in $\FF_{2^8}$ due to the absence of trinomial irreducible polynomial of degree $8$ over $\FF_{2}$~\cite[Theorem~2]{BKL2016}. However, it is possible to have elements with XOR count of $1$ over rings.

	Consider the binary matrix $C=[[2],[3],[4],[5],[6],[7],[8],[1,3]]$ which is the companion matrix of $x^8+x^2+1$ over $\FF_{2}$. If we replace $\alpha$ by $C$, then the matrices $A_1$ and $A_2$ in (\ref{Equ_5-nMDS_over_field_F_2^4}) and $B_1$ and $B_2$ in (\ref{Equ_6-nMDS_over_field_F_2^4}) will be $4$-NMDS, $5$-NMDS, $5$-NMDS, and $6$-NMDS over $GL(8,\FF_{2})$, respectively. In addition, the implementation cost of $C$ and $C^{-1}$ is $1$ XOR. Hence, the implementation cost of $A_1$, $A_2$, $B_1$ and $B_2$ over $GL(8,\FF_{2})$ are $26,25,26$ and $25$ XORs, respectively.
\end{remark}

\begin{remark}
	Over the field $\mathbb{F}_{2^4}$, the matrix $B_2$ in (\ref{Equ_6-nMDS_over_field_F_2^4}) has the lowest XOR count among the $6$-NMDS matrices of order $6$ and $\mathcal{K}=3$.
\end{remark}

\subsection{Construction of $7\times 7$ Recursive NMDS matrices}\label{Section_7-near-MDS_GDLS}
In this section, we propose three GDLS matrices of order $7$ that yield NMDS matrices over the field $\mathbb{F}_{2^4}$ for $\mathcal{K}=4$. Consider the GDLS matrices $B_1,~B_2$ and $B_3$ of order $7$ which are constructed as follows:
\begin{enumerate}[(i)]
	\item $B_1$:~$\rho_1=[6,7,4,5,2,3,1],~\rho_2=[3,2,1,4,7,6,5],~D_1=diag(\splitatcommas{1, 1, 1, 1, 1, 1, \alpha})$ and $D_2=diag(1, 0, \alpha^2, 0, \alpha, 0, \alpha^2)$
	\item $B_2$:~$\rho_1=[7, 1, 2, 3, 4, 5, 6],~\rho_2=[6, 7, 5, 4, 1, 3, 2],~D_1=diag(\splitatcommas{1, \alpha^{-1}, \alpha, 1, \alpha, \alpha, 1})$ and $D_2=diag(0, 1, 0, 1, 0, 1, 1)$
	\item $B_3$:~$\rho_1=[7, 1, 2, 3, 4, 5, 6],~\rho_2=[5, 2, 6, 7, 3, 1, 4],~D_1=diag(\splitatcommas{1, 1, 1, 1, 1, 1, 1})$ and $D_2=diag(0, \alpha^{-1}, 0, 1, 0, \alpha^{-1}, 1)$
\end{enumerate}

\begin{equation}\label{Equ_7-nMDS_over_field_F_2^4}
	\begin{aligned}
		B_1&=
		\begin{bmatrix}
			0 & 0 & \alpha^{2} & 0 & 0 & 0 & \alpha \\
			0 & 0 & 0 & 0 & 1 & 0 & 0 \\
			1 & 0 & 0 & 0 & 0 & 1 & 0 \\
			0 & 0 & 1 & 0 & 0 & 0 & 0 \\
			0 & 0 & 0 & 1 & 0 & 0 & \alpha^{2} \\
			1 & 0 & 0 & 0 & 0 & 0 & 0 \\
			0 & 1 & 0 & 0 & \alpha & 0 & 0
		\end{bmatrix}
		B_2=
		\begin{bmatrix}
			0 & \alpha^{-1} & 0 & 0 & 0 & 0 & 0 \\
			0 & 0 & \alpha & 0 & 0 & 0 & 1 \\
			0 & 0 & 0 & 1 & 0 & 1 & 0 \\
			0 & 0 & 0 & 1 & \alpha & 0 & 0 \\
			0 & 0 & 0 & 0 & 0 & \alpha & 0 \\
			0 & 0 & 0 & 0 & 0 & 0 & 1 \\
			1 & 1 & 0 & 0 & 0 & 0 & 0
		\end{bmatrix}
		B_3=
		\begin{bmatrix}
			0 & 1 & 0 & 0 & 0 & \alpha^{-1} & 0 \\
			0 & \alpha^{-1} & 1 & 0 & 0 & 0 & 0 \\
			0 & 0 & 0 & 1 & 0 & 0 & 0 \\
			0 & 0 & 0 & 0 & 1 & 0 & 1 \\
			0 & 0 & 0 & 0 & 0 & 1 & 0 \\
			0 & 0 & 0 & 0 & 0 & 0 & 1 \\
			1 & 0 & 0 & 1 & 0 & 0 & 0
		\end{bmatrix},
	\end{aligned}
\end{equation}
where $\alpha$ is a primitive element of~$\FF_{2^4}$ and a root of $x^4+x+1$. 
It can be verified that matrix $B_1$ is a $5$-NMDS matrix with an XOR count of $(1+2+1+2) + 4\cdot 4 = 22$, $B_2$ is a $6$-NMDS matrix with an XOR count of $(1+1+1+1) + 4\cdot 4 = 20$, and $B_3$ is a $7$-NMDS matrix with an XOR count of $(1+1) + 4\cdot 4 = 18$.

\begin{remark}\label{Example_k-nMDS_order_7_ring_F_2^8}
	If we replace $\alpha$ by $C$ (the binary matrix in Remark~\ref{Example_k-nMDS_order_5_6_ring_F_2^8}), then the matrices $B_1,~B_2$ and $B_3$ in (\ref{Equ_7-nMDS_over_field_F_2^4}) will be $5$-NMDS, $6$-NMDS and $7$-NMDS over $GL(8,\FF_{2})$, respectively. The binary matrix $C^2$ can be implemented with $2$ XORs. Hence, $B_1,~B_2$ and $B_3$ can be implemented with $38,36$ and $34$ XORs, respectively, over $GL(8,\FF_{2})$.
\end{remark}

\begin{remark}\label{Remark_5-nMDS_order_7}
	We know that in a DLS matrix $M=DLS(\rho_1;D_1,D_2)$ of order $n\geq 2$, if $D_2$ is singular, then $M^k$ cannot be NMDS for $k\leq n-2$. However, the result is not true for GDLS matrices. For example the matrix $B_1$ of order 7 in (\ref{Equ_7-nMDS_over_field_F_2^4}) is $5$-NMDS.
\end{remark}
%

\subsection{Construction of $8\times 8$ Recursive NMDS matrices}\label{Section_8-near-MDS_GDLS}
As $4$ and $8$ are the most commonly used diffusion layer matrix sizes, we look for a $k$-NMDS GDLS matrix of order $8$ over $\FF_{2^4}$. However, we were unable to find a GDLS matrix of order $8$, which corresponds to $7$-NMDS or $8$-NMDS. We have proposed two GDLS matrices of order $8$ that yield NMDS matrices over the field $\mathbb{F}_{2^8}$ with $\mathcal{K}=4$. Consider the GDLS matrices $B_1$ and $B_2$ of order $8$ which are constructed as follows:
\begin{enumerate}[(i)]
	\item $B_1:~\rho_1=[2, 3, 4, 5, 6, 7, 8, 1],~\rho_2=[3, 8, 5, 2, 1, 4, 6, 7],~D_1=diag(\splitatcommas{1, 1, 1, 1, 1, \alpha^{-2}, 1, 1})$ and $D_2=diag(1, 0, \alpha, 0, 1, 0, \alpha^{-1}, 0)$
	\item $B_2:~\rho_1=[2, 3, 4, 5, 6, 7, 8, 1],~\rho_2=[3, 8, 5, 2, 1, 4, 6, 7],~D_1=diag(\splitatcommas{1, 1, 1, \alpha^2, 1, 1, 1, 1})$ and $D_2=diag(\alpha, 0, 1, 0, 1, 0, 1, 0)$
\end{enumerate}

\begin{equation}\label{Equ_8-nMDS_over_field_F_2^4}
	\begin{aligned}
		B_1&=
		\begin{bmatrix}
			0 & 0 & 0 & 0 & 1 & 0 & 0 & 1 \\
			1 & 0 & 0 & 0 & 0 & 0 & 0 & 0 \\
			1 & 1 & 0 & 0 & 0 & 0 & 0 & 0 \\
			0 & 0 & 1 & 0 & 0 & 0 & 0 & 0 \\
			0 & 0 & \alpha & 1 & 0 & 0 & 0 & 0 \\
			0 & 0 & 0 & 0 & 1 & 0 & \alpha^{-1} & 0 \\
			0 & 0 & 0 & 0 & 0 & \alpha^{-2} & 0 & 0 \\
			0 & 0 & 0 & 0 & 0 & 0 & 1 & 0
		\end{bmatrix}~~
		B_2=
		\begin{bmatrix}
			0 & 0 & 0 & 0 & 1 & 0 & 0 & 1 \\
			1 & 0 & 0 & 0 & 0 & 0 & 0 & 0 \\
			\alpha & 1 & 0 & 0 & 0 & 0 & 0 & 0 \\
			0 & 0 & 1 & 0 & 0 & 0 & 0 & 0 \\
			0 & 0 & 1 & \alpha^{2} & 0 & 0 & 0 & 0 \\
			0 & 0 & 0 & 0 & 1 & 0 & 1 & 0 \\
			0 & 0 & 0 & 0 & 0 & 1 & 0 & 0 \\
			0 & 0 & 0 & 0 & 0 & 0 & 1 & 0
		\end{bmatrix},
	\end{aligned}
\end{equation}
where $\alpha$ is a primitive element of~$\FF_{2^8}$ and a root of $x^8+x^7+x^6+x+1$. The matrix $B_1$ is a $7$-NMDS matrix with a XOR count of $(4+3+3)+4\cdot 8=42$ and $B_2$ is a $8$-NMDS matrix with a XOR count of $(3+4)+4\cdot 8=39$.

\begin{remark}
	Consider the binary matrix $C_{8}=[[8],[1,2],[2,8],[3],[4],[5],[6],[7]]$ whose minimal polynomial is $x^8 + x^7 + x^2 + x + 1$. Then by replacing $\alpha$ by $C_{8}$, the matrices $B_1$, and $B_2$ in (\ref{Equ_8-nMDS_over_field_F_2^4}) will be $7$-NMDS and $8$-NMDS over $GL(8,\FF_{2})$, respectively. In addition, the implementation cost of $C_{8}$ is $2$ XORs. Also, $C_{8}^{-1}$, $C_{8}^{2}$ and $C_{8}^{-2}$ can be implemented with $2$, 4 and $4$ XORs respectively. Hence, $B_1$ and $B_2$ can be implemented with $40$ and $38$ XORs, respectively, over $GL(8,\FF_{2})$.
\end{remark}

\begin{table}
	\centering
	\caption{Comparison of recursive NMDS matrices of order $n$.}\label{table_comparison_k-NMDS}
	\begin{tabular}{|cccccc|}
	\hline
	Order $n$~~ & Input &~ Iterations &~~ Field/Ring  &~~ XOR count 	&~~ References \\ \hline
		$4$	& 4-bit & 34 & $M_4(\mathbb{F}_2)$ & 	\textbf{1} 	& \cite{Recursive_nMDS_2021} \\
		$4$	& 4-bit & 16 & $M_4(\mathbb{F}_2)$ & 2 & \cite{Recursive_nMDS_2021} \\
		$4$	& 4-bit & 10 & $M_4(\mathbb{F}_2)$ & 3 & \cite{Recursive_nMDS_2021} \\
		$4$	& 4-bit & 7  & $M_4(\mathbb{F}_2)$ & 4 & \cite{Recursive_nMDS_2021} \\
		$4$	& 4-bit & 5  & $M_4(\mathbb{F}_2)$ & 7 & \cite{Recursive_nMDS_2021} \\
		$4$	& 4-bit & 3  & $M_4(\mathbb{F}_2)$ & 8 & \cite{Recursive_nMDS_2021} \\
		$4$ & 4-bit & 3 & $\FF_{2^4}$ &	$8$	& Section~\ref{Section_4-near-MDS_GDLS}\\
		$4$ & 4-bit & 3 & $\FF_{2^4}$ &	$8$	& \cite{SM19}\\
		$4$	& 4-bit & \textbf{2} & $M_4(\mathbb{F}_2)$ & 12 & \cite{Recursive_nMDS_2021} \\ \hline
		$4$	& 8-bit & 66 & $M_8(\mathbb{F}_2)$ & 	\textbf{1} 	& \cite{Recursive_nMDS_2021} \\
		$4$	& 8-bit & 34 & $M_8(\mathbb{F}_2)$ & 2 & \cite{Recursive_nMDS_2021} \\
		$4$	& 8-bit & 16 & $M_8(\mathbb{F}_2)$ & 4 & \cite{Recursive_nMDS_2021} \\
		$4$	& 8-bit & 10 & $M_8(\mathbb{F}_2)$ & 6 & \cite{Recursive_nMDS_2021} \\
		$4$	& 8-bit & 7  & $M_8(\mathbb{F}_2)$ & 8 & \cite{Recursive_nMDS_2021} \\
		$4$	& 8-bit & 3  & $M_8(\mathbb{F}_2)$ & 16 & \cite{Recursive_nMDS_2021} \\
		$4$ & 8-bit & 3 & $\FF_{2^8}$ &	$16$	& Section~\ref{Section_4-near-MDS_GDLS}\\
		$4$ & 8-bit & 3 & $\FF_{2^8}$ &	$16$	& \cite{SM19}\\
		$4$	& 8-bit & \textbf{2} & $M_4(\mathbb{F}_2)$ & 24 & \cite{Recursive_nMDS_2021} \\ \hline \hline
		$5$	& 4-bit & 86 & $M_8(\mathbb{F}_2)$ & 	\textbf{1} 	& \cite{Recursive_nMDS_2021} \\
		$5$	& 4-bit & 46 & $M_8(\mathbb{F}_2)$ & 	2 	& \cite{Recursive_nMDS_2021} \\
		$5$	& 4-bit & 20 & $M_8(\mathbb{F}_2)$ & 	3 	& \cite{Recursive_nMDS_2021} \\
		$5$	& 4-bit & 15 & $M_8(\mathbb{F}_2)$ & 	4 	& \cite{Recursive_nMDS_2021} \\
		$5$	& 4-bit & 8 & $M_8(\mathbb{F}_2)$ & 	8 	& \cite{Recursive_nMDS_2021} \\
		$5$ & 4-bit & 5 & $\FF_{2^4}/0$x$13$ &	$13$	& Section~\ref{Section_5-near-MDS_GDLS}\\
		$5$ & 4-bit & \textbf{4} & $\FF_{2^4}/0$x$13$ &	$14$	& Section~\ref{Section_5-near-MDS_GDLS}\\ \hline
		$5$	& 8-bit & 120 & $M_8(\mathbb{F}_2)$ & 	\textbf{1} 	& \cite{Recursive_nMDS_2021} \\
		$5$	& 8-bit & 86 & $M_8(\mathbb{F}_2)$ & 	2 	& \cite{Recursive_nMDS_2021} \\
		$5$	& 8-bit & 46 & $M_8(\mathbb{F}_2)$ & 	4 	& \cite{Recursive_nMDS_2021} \\
		$5$	& 8-bit & 20 & $M_8(\mathbb{F}_2)$ & 	6 	& \cite{Recursive_nMDS_2021} \\
		$5$	& 8-bit & 15 & $M_8(\mathbb{F}_2)$ & 	8 	& \cite{Recursive_nMDS_2021} \\
		$5$	& 8-bit & 8 & $M_8(\mathbb{F}_2)$ & 	16 	& \cite{Recursive_nMDS_2021} \\
		$5$ & 8-bit & 5 & $GL(8,\mathbb{F}_2)$ &	$25$  & Remark~\ref{Example_k-nMDS_order_5_6_ring_F_2^8}\\
		$5$ & 8-bit & \textbf{4} & $GL(8,\mathbb{F}_2)$ &	$26$& Remark~\ref{Example_k-nMDS_order_5_6_ring_F_2^8}\\ \hline \hline
		$6$ & 4-bit & 6 & $\FF_{2^4}/0$x$13$ &	\textbf{13}	& Section~\ref{Section_6-near-MDS_GDLS}\\ 
		$6$ & 4-bit & \textbf{5} & $\FF_{2^4}/0$x$13$ &	$14$	& Section~\ref{Section_6-near-MDS_GDLS}\\ \hline
		$6$ & 8-bit & 6 & $GL(8,\mathbb{F}_2)$ & \textbf{25} & Remark~\ref{Example_k-nMDS_order_5_6_ring_F_2^8}\\
		$6$ & 8-bit & \textbf{5} & $GL(8,\mathbb{F}_2)$ &	$26$& Remark~\ref{Example_k-nMDS_order_5_6_ring_F_2^8}\\ \hline \hline
		$7$ & 4-bit & 7 & $\FF_{2^4}/0$x$13$ &	\textbf{18}	& Section~\ref{Section_7-near-MDS_GDLS}\\ 
		$7$ & 4-bit & 6 & $\FF_{2^4}/0$x$13$ &	20	& Section~\ref{Section_7-near-MDS_GDLS}\\ 
		$7$ & 4-bit & \textbf{5} & $\FF_{2^4}/0$x$13$ &	$22$	& Section~\ref{Section_7-near-MDS_GDLS}\\ \hline
		$7$ & 8-bit & 7 & $GL(8,\mathbb{F}_2)$ &	\textbf{34}	& Remark~\ref{Example_k-nMDS_order_7_ring_F_2^8}\\ 
		$7$ & 8-bit & 6 & $GL(8,\mathbb{F}_2)$ &	36	& Remark~\ref{Example_k-nMDS_order_7_ring_F_2^8}\\ 
		$7$ & 8-bit & \textbf{5} & $GL(8,\mathbb{F}_2)$ &	38	& Remark~\ref{Example_k-nMDS_order_7_ring_F_2^8}\\ \hline \hline
		$8$ & 8-bit & 8 & $\FF_{2^8}/0$x$1c3$ & 39	& Section~\ref{Section_8-near-MDS_GDLS}\\ 
		$8$ & 8-bit & 8 & $GL(8,\mathbb{F}_2)$ &	\textbf{38}	& Section~\ref{Section_8-near-MDS_GDLS}\\ 
		$8$ & 8-bit & \textbf{7} & $\FF_{2^8}/0$x$1c3$ & 42 & Section~\ref{Section_8-near-MDS_GDLS}\\ 
		$8$ & 8-bit & \textbf{7} & $GL(8,\mathbb{F}_2)$ & 40 & Section~\ref{Section_8-near-MDS_GDLS}\\ \hline
	\end{tabular}
\end{table}
\noindent Until now, we have discussed NMDS matrices in a recursive setup. While these matrices have a low hardware cost, they do require some clock cycles. To use recursive NMDS (say, $k$-NMDS) matrices in an unrolled implementation, we have to add $k$ copies of the matrix to the circuit in sequence, which may increase the cost of the diffusion layer. Thus, they might not be suitable for block ciphers such as PRINCE~\cite{PRINCE}, MANTIS~\cite{SKINNY}, and FUTURE~\cite{FUTURE}, which operate within a single clock cycle. From the next section on, we will discuss nonrecursive constructions of NMDS matrices.

\section{Construction of nonrecursive NMDS matrices}\label{Section_Single_clock_NMDS}
The construction of nonrecursive MDS matrices is typically based on specific matrix types such as circulant matrices, Hadamard matrices, Cauchy matrices, Vandermonde matrices, and Toeplitz matrices. A brief summary of such constructions is presented in~\cite{MDS_Survey}. Circulant and Hadamard matrices of order $n$ can have at most $n$ distinct elements; thus, these matrices are used to reduce the search space. Also, circulant matrices have the flexibility to be implemented in both round-based and serialized implementations~\cite{CYCLICM}. In~\cite{Li_Wang_2017}, the authors have studied the construction of NMDS matrices using circulant and Hadamard matrices and present some generic NMDS matrices of order $n$ for the range of $5\leq n \leq 9$.

\begin{definition}\label{CIRCULANT}
    An $n\times n$ matrix $M$ is said to be a right circulant (or circulant) matrix if its elements are determined by the elements of its first row $x_1,x_2,\ldots,x_n$ as
	$$M=Circ(x_1,x_2,\ldots,x_{n})
	=\begin{bmatrix}
	x_1 & x_2 & \ldots & x_{n}\\
	x_{n} & x_1 & \ldots & x_{n-1}\\
	\vdots & \vdots &\vdots &\vdots \\
	x_2 & x_3 & \ldots & x_1\\
	\end{bmatrix}.$$
\end{definition}

In the context of implementing block ciphers, we know that if an efficient matrix $M$ used in encryption is involutory, then its inverse $M^{-1}=M$ applied for decryption will also be efficient. Therefore, it is particularly important to locate NMDS matrices that are also involutory. In this regard, Li et al.~\cite{Li_Wang_2017} show that for $n > 4$, no circulant matrices of order $n$ over $\mathbb{F}_{2^r}$ can simultaneously be involutory and NMDS. We recall it in the following theorem.

\begin{theorem}\cite{Li_Wang_2017}\label{Th_circulant_NMDS_involutory}
    Over the field $\mathbb{F}_{2^r}$, circulant involutory matrices of order $n> 4$ are not NMDS.
\end{theorem}

\begin{remark}\label{Remark_Circulant_NMDS_Involutory}
	For $n<4$, there may exist circulant involutory NMDS matrices over $\mathbb{F}_{2^r}$. For example, the circulant matrix $Circ(0,1,1,1)$ of order $4$ is both involutory and NMDS over the field $\mathbb{F}_{2^r}$.
\end{remark}

\begin{remark}
    According to~\cite{MDS_Survey,GR15}, the above result is also true for circulant MDS matrices of order $n$ with a modified lower bound of $n\geq 3$.
\end{remark}
For symmetric cryptography, having an orthogonal matrix as the linear diffusion layer simplifies decryption because the transpose of an orthogonal matrix is its inverse. This makes orthogonal matrices ideal for constructing the linear diffusion layer. Matrices of order $2^n$ are particularly important in cryptography. However, as stated in~\cite[Lemma~5]{GR15}, for $n \geq 2$, we know that any orthogonal circulant matrix of order $2^n$ over the field $\mathbb{F}_{2^r}$ is not MDS. But circulant MDS matrices of different orders may be orthogonal over $\mathbb{F}_{2^r}$~\cite[Remark~24]{MDS_Survey}.

\begin{remark}\label{Remark_circulant_NMDS_orthogonal}
	NMDS circulant orthogonal matrices of any order may exist over the field $\mathbb{F}_{2^r}$. For example, consider the circulant matrices $Circ(0, \alpha^{3} + \alpha + 1, \alpha^{3} + \alpha^{2} + \alpha, \alpha^{3} + 1, \alpha^{3} + \alpha^{2} + 1)$, $Circ(0, 1, \alpha, \alpha^{2} + \alpha + 1, \alpha^{3} + \alpha + 1, \alpha^{3} + \alpha^{2} + \alpha)$, and $Circ(0, 1, \alpha, \alpha + 1, \alpha + 1, \alpha^{3} + \alpha^{2} + \alpha + 1, 1, \alpha^{3} + \alpha^{2})$ of order $5,6$, and $8$, respectively, where $\alpha$ is a primitive element of~$\FF_{2^4}$ and a root of the polynomial $x^4+x+1$. It can be checked that these matrices are both orthogonal and NMDS.
\end{remark}

In~\cite{CYCLICM}, the authors suggest a new category of matrices known as left-circulant matrices. These matrices retain the advantages of circulant matrices.

\begin{definition}
    An $n\times n$ matrix $M$ is said to be a left-circulant matrix if each successive row is obtained by a left shift of the previous row i.e.
    $$M=l\text{-}Circ(x_1,x_2,\ldots,x_{n})
	=\begin{bmatrix}
	x_1 & x_2 & \ldots & x_{n}\\
	x_{2} & x_3 & \ldots & x_{1}\\
	\vdots & \vdots &\vdots &\vdots \\
	x_n & x_1 & \ldots & x_{n-1}\\
	\end{bmatrix}.$$
\end{definition}

Note that a left-circulant matrix is symmetric; consequently, if the matrix is orthogonal, then it is involutory, and vice versa.
\begin{remark}\label{Remark_left_circulant_involutory_NMDS}
	From Lemma~\ref{Lemma_same_branch_perm}, we know that if $M$ is an NMDS matrix, then for any permutation matrix $P$, $PM$ is also an NMDS matrix. Additionally, as per Remark \ref{Remark_circulant_NMDS_orthogonal}, it is possible to obtain a circulant NMDS matrix $M=Circ(x_1, x_2,\ldots,x_{n})$ of any order $n$ over $\mathbb{F}_{2^r}$ that is orthogonal. Now, consider the permutation matrix $P$ of order $n$ as follows:
	\begin{center}
		$P=~\begin{bmatrix}
		1 & 0 & 0  & \ldots & 0 & 0 & 0\\
		0 & 0 & 0 & \ldots & 0 & 0 & 1\\
		0 & 0 & 0 &\ldots & 0 &  1 & 0\\
		0 & 0 & 0 & \ldots & 1 & 0 & 0\\
		\vdots &\vdots &\vdots & \ldots & \vdots & \vdots & \vdots\\
		0 & 1 & 0 & \ldots & 0 & 0 & 0\\
		\end{bmatrix}.$
	\end{center}
	It can be easily verified that $PM$ =~$l$-$Circ(x_1,x_2,\ldots,x_{n})$. Also, since $M$ is orthogonal and $P$ is a permutation matrix, we have $$(PM)^T=M^TP^T=M^{-1}P^{-1}=(PM)^{-1}.$$
	Thus, pre-multiplying $M$ with $P$ will not alter its NMDS property and orthogonality. Consequently, the resulting matrix $PM=l$-$Circ(x_1,x_2,\ldots,x_{n})$ will be both orthogonal and NMDS, making it an involutory NMDS matrix. Therefore, NMDS left-circulant involutory (orthogonal) matrices of any order may exist over the field $\mathbb{F}_{2^r}$.
\end{remark}

\begin{definition}
	A $2^n \times 2^n$ matrix $M$ in $\mathbb{F}_{2^r}$ is considered a Hadamard matrix if it can be written in the following form:
    $$
    M=\begin{bmatrix}
    H_1 & H_2 \\
    H_2 & H_1
    \end{bmatrix}
    $$
    where $H_1$ and $H_2$ are also Hadamard.
\end{definition}

The most significant advantage of Hadamard matrices is the potential for constructing involutory matrices. If the elements of the matrix are chosen so that the first row sums to one, the resulting matrix will be involutory~\cite{MDS_Survey}.

\noindent The absence of any zero entries is a necessary condition for matrices such as Hadamard, circulant and left-circulant matrices to be MDS. Therefore, these matrices result in a high implementation cost due to $\mathcal{K}=n(n-1)$. Having zero entries (with a maximum of one zero per row or column) does not affect the NMDS property of these matrices, leading to a low implementation cost with $\mathcal{K}=n(n-2)$. Taking advantage of this, the authors in~\cite{Li_Wang_2017} provided some generic lightweight involutory NMDS matrices of order $8$ from Hadamard matrices.

\begin{theorem}
	For a Hadamard, circulant, or left-circulant NMDS matrix of order $n$ over $\mathbb{F}_{2^r}$ with $n\geq 5$, the XOR count is at least $XOR(\beta)\cdot n+n(n-2)\cdot r$ over the field $\mathbb{F}_{2^r}$, where $\beta$ ($\neq 1$) is a nonzero element in $\mathbb{F}_{2^r}$ with the lowest XOR count value in that field.
\end{theorem}
\begin{proof}
	An NMDS matrix $B$ of order $n$ has branch number of $n$. Therefore, according to Theorem~\ref{Th_max_branch_number_binary_matrix}, a matrix of order $n$ with $n \geq 5$ and entries from the set $\set{0,1}\subseteq \mathbb{F}_{2^r}$ cannot be NMDS. This means that $B$ must contain an element $\gamma \not \in \set{0,1}$. Additionally, for an NMDS matrix, we must have $\mathcal{K} \geq n(n-2)$. Also, each row in a Hadamard, circulant, or left-circulant matrix is a rearrangement of the first row. Hence, for these matrices to be NMDS over the field $\mathbb{F}_{2^r}$, the minimum XOR count must be $XOR(\beta)\cdot n+n(n-2)\cdot r$.\qed
\end{proof}

\noindent The lowest XOR count value (of an element) in the field $\mathbb{F}_{2^4}$ is one, which allows us to obtain the lowest possible XOR count of Hadamard, circulant, or left-circulant NMDS matrices of various orders over $\mathbb{F}_{2^4}$ as shown in Table~\ref{Table_lowest_XOR_NMDS}.

\begin{table}
	\centering
	\caption{Lowest possible XOR count of Hadamard, circulant, or left-circulant NMDS matrices of order $n$ over~$\mathbb{F}_{2^4}$.}\label{Table_lowest_XOR_NMDS}
	\begin{tabular}{|c||cccc|}
		\hline
		order $n$ & ~5~ & ~6~ & ~7~ & ~8~ \\ \hline
		Lowest XOR count &~ 65~ & ~102~ & ~147~ & ~200~ \\ \hline
	\end{tabular}
\end{table}
\noindent The use of Toeplitz matrices for the construction of MDS matrices has been explored in the literature~\cite{XORM,PXOR2}, and we will discuss them for the construction of NMDS matrices.
\begin{definition}\label{def:toeplitz_NMDS}
	The $n\times n$ matrix 
	\begin{center}
		\begin{equation*}
		M=\begin{bmatrix}
		x_1 & x_2 & x_3 & \ldots & x_{n-2} & x_{n-1} & x_{n}\\
		y_1 & x_1 & x_2 & \ldots & x_{n-3} & x_{n-2} & x_{n-1}\\
		y_2 & y_1 & x_1 & \ldots & x_{n-4} & x_{n-3} & x_{n-2}\\
		\vdots & \vdots & \vdots & \vdots & \vdots & \vdots & \vdots\\
		y_{n-1} & y_{n-2} & y_{n-3} & \ldots & y_{2} & y_{1} & x_1
		\end{bmatrix}
		\end{equation*}
	\end{center}
	is called a Toeplitz matrix of order $n$.
\end{definition}

\noindent It is easy to check that circulant matrices are a special type of Toeplitz matrices. Also, like circulant matrices, it is not possible for Toeplitz matrices of order $n>4$ to be both NMDS and involutory over a field of characteristic 2.

\begin{theorem}\label{Th_Toeplitz_NMDS_involutory}
    Over the field $\mathbb{F}_{2^r}$, Toeplitz involutory matrices of order $n> 4$ are not NMDS.
\end{theorem}
\begin{proof}
	Let $M$ be a Toeplitz matrix (as in Definition~\ref{def:toeplitz_NMDS}) of order $n$ which is both involutory and NMDS over the field $\mathbb{F}_{2^r}$, where $n>4$. We will examine two scenarios: when $n$ is even and when $n$ is odd.
	
	\vspace{4pt}
	\noindent
	\textbf{Case 1:} $n$ is even.\newline
	In an NMDS matrix, there may be a zero entry. So this case splits into two subcases: $x_n\neq 0$ and $x_n= 0.$
	
	\noindent \textbf{Case 1.1:}~When $x_n\neq 0$.\newline
	The $(n-1)$-th element in the 1st row of $M^2$ is
	\begin{equation*}
	    \begin{aligned}
        	(M^2)_{1,n-1} &= M_{row(1)}\cdot M_{column(n-1)}\\
        	&=x_{1}x_{n-1}+x_{2}x_{n-2}+\cdots +x_{\frac{n}{2}}x_{\frac{n}{2}}+\cdots +x_{n-1}x_{1}+x_{n}y_{1}\\
        	&=x_{\frac{n}{2}}^2+x_{n}y_{1}.
	    \end{aligned}
	\end{equation*}
	Since $M$ is involutory, we have $(M^2)_{1,n-1}=0$. Therefore, from above we have 
	\begin{align}
	    & x_{\frac{n}{2}}^2+x_{n}y_{1}=0 \label{Eqn_proof_Toeplitz_1a} \\
	    \implies & y_{1}= x_{\frac{n}{2}}^2x_{n}^{-1} \label{Eqn_proof_Toeplitz_1b}.
	\end{align}
	We have
	\begin{equation*}
	    \begin{aligned}
        	(M^2)_{1,n-2} &= M_{row(1)}\cdot M_{column(n-2)}\\
        	&=x_{1}x_{n-2}+x_{2}x_{n-3}+\cdots +x_{\frac{n-2}{2}}x_{\frac{n}{2}}+x_{\frac{n}{2}}x_{\frac{n-2}{2}}+\cdots\\
        	&~~~~+x_{n-3}x_{2}+x_{n-2}x_{1}+x_{n-1}y_{1}+x_{n}y_{2}\\
        	&=x_{n-1}y_{1}+x_{n}y_{2}.
	    \end{aligned}
	\end{equation*}
	Also, $(M^2)_{1,n-2}=0$, which results in
	\begin{equation}\label{Eqn_proof_Toeplitz_2}
    	\begin{aligned}
    	    & x_{n-1}y_1+x_{n}y_{2}=0
    	\end{aligned}
	\end{equation}
	Now, from Equation~\ref{Eqn_proof_Toeplitz_1b} and Equation~\ref{Eqn_proof_Toeplitz_2}, we have
	\begin{equation}\label{Eqn_proof_Toeplitz_3}
    	\begin{aligned}
    	    y_{2} &= x_{\frac{n}{2}}^2x_{n-1}x_{n}^{-2}.
    	\end{aligned}
	\end{equation}
	Also, from $(M^2)_{3,n-1}=0$, we have
	\begin{align}
	    &x_{\frac{n-2}{2}}^2+x_{n-1}y_{2} =0 \notag \\
	    \implies & x_{\frac{n-2}{2}}^2 + x_{n-1} \cdot x_{\frac{n}{2}}^2 x_{n-1} x_{n}^{-2} =0 && [\text{From Equation~\ref{Eqn_proof_Toeplitz_3}}] \notag \\
	    \implies & x_{\frac{n-2}{2}}^2x_{n}^{2} = x_{\frac{n}{2}}^2 x_{n-1}^{2} \notag \\
	    \implies & x_{\frac{n-2}{2}} x_{n} = x_{\frac{n}{2}} x_{n-1} && [\text{Since characteristic of $\mathbb{F}_{2^r}$ is 2}] \label{Eqn_proof_Toeplitz_4}
	\end{align}
    Now consider the input vector $v=[0,0,\ldots,\underbrace{x_{n}}_{\text{$\frac{n}{2}$-th}},0,\ldots,x_{\frac{n}{2}}]^T$ of $M$. Therefore, we have
	\begin{equation*}
		\begin{aligned}
			M\cdot v &= [x_{\frac{n}{2}}x_n+x_n x_{\frac{n}{2}},~x_{\frac{n-2}{2}} x_{n} + x_{n-1} x_{\frac{n}{2}},~*,\ldots,*,\underbrace{y_1x_n+x_{\frac{n}{2}}^2}_{\text{$(\frac{n}{2}+1)$-th}},*,\ldots,*]^T \notag \\
        &= [0,0,*,\ldots,\underbrace{0}_{\text{$(\frac{n}{2}+1)$-th}},*,\ldots,*]^T,
		\end{aligned}
	\end{equation*}
    where $*$ denotes some entry may or may not be zero. Here the second and $(\frac{n}{2}+1)$-th coordinates of $M\cdot v$ are zero by Equation~\ref{Eqn_proof_Toeplitz_4} and Equation~\ref{Eqn_proof_Toeplitz_1a}, respectively. Thus, the sum of nonzero elements of input vector ($v$) and output vector ($M\cdot v$) is $\leq 2+(n-3)<n$ i.e. branch number of $M<n$. This contradicts that $M$ is NMDS.
    
	\noindent \textbf{Case 1.2:}~When $x_n= 0$.\newline
	If $x_n= 0$, then from Equation~\ref{Eqn_proof_Toeplitz_1a}, we conclude that $x_{\frac{n}{2}}^2=0$ which implies $x_{\frac{n}{2}}=0$. Therefore, the Toeplitz matrix $M$ has two zero entries in its first row, which contradicts the fact that $M$ is NMDS.
	
	\vspace{4pt}
	\noindent
	\textbf{Case 2:} $n$ is odd.\newline
	The $n$-th element in the 1st row of $M^2$ is
	\begin{equation*}
	    \begin{aligned}
        	(M^2)_{1,n} &= M_{row(1)}\cdot M_{column(n)}\\
        	&=x_{1}x_{n}+x_{2}x_{n-1}+\cdots +x_{\frac{n+1}{2}}x_{\frac{n+1}{2}}+\cdots +x_{n-1}x_{2}+x_{n}x_{1}\\
        	&=x_{\frac{n+1}{2}}^2.
	    \end{aligned}
	\end{equation*}
	Also, we have $(M^2)_{2,n-1}=x_{\frac{n-1}{2}}^2$.
	Therefore, since $M$ is involutory, it follows that $(M^2)_{1,n}=(M^2)_{2,n-1}=0$, implying that $x_{\frac{n-1}{2}}=x_{\frac{n+1}{2}}=0$. This means that $M$ has two zero entries in its first row, which contradicts that $M$ is an NMDS matrix. Hence, the proof.\qed
\end{proof}

\begin{remark}
	Circulant matrices are a particular type of Toeplitz matrices, and thus, from Remark~\ref{Remark_Circulant_NMDS_Involutory}, we can say that for $n<4$, there may exist Toeplitz involutory NMDS matrices over $\mathbb{F}_{2^r}$.
\end{remark}

\begin{remark}
    From~\cite[Theorem~2]{XORM}, we know that for $n \geq 2$, any orthogonal Toeplitz matrix of order $2^n$ over the field $\mathbb{F}_{2^r}$ is not MDS. However, this result does not hold for NMDS matrices. Circulant matrices are a particular type of Toeplitz matrices, and thus, from Remark~\ref{Remark_circulant_NMDS_orthogonal}, we can say that Toeplitz orthogonal NMDS matrices of any order may exist over the field $\mathbb{F}_{2^r}$.
\end{remark}

\noindent Hankel matrices, introduced in~\cite{MDS_Survey} for MDS matrix construction, are similar to Toeplitz matrices in that each ascending skew diagonal from left to right is constant.
\begin{definition}
	The $n\times n$ matrix \begin{center}
		\begin{equation*}
		M=\begin{bmatrix}
		x_{1}~ & x_{2}~ & x_{3}~ & \ldots & x_{n-1} & x_{n}\\
		x_{2}~ & x_{3}~ & x_{4}~ & \ldots & x_{n} & y_{1}\\
		x_{3}~ & x_{4}~ & x_{5}~ & \ldots & y_{1} & y_{2}\\
		\vdots~ & \vdots~ & \vdots~ & \vdots & \vdots & \vdots\\
		x_{n}~ & y_{1}~ & y_{2}~ & \ldots & y_{n-2} & y_{n-1}
		\end{bmatrix}
		\end{equation*}
	\end{center}
	is called a Hankel matrix.
\end{definition}

It is important to note that a left-circulant matrix is a special case of Hankel matrix. Hankel matrices are symmetric and may be described by their first row and last column. Thus an involutory (orthogonal) Hankel matrix is orthogonal (involutory).

\begin{remark}
    From~\cite[Theorem~7.5]{MDS_Survey}, we know that for $n \geq 2$, any involutory (orthogonal) Hankel matrix of order $2^n$ over the field $\mathbb{F}_{2^r}$ is not MDS. However, this result does not hold for NMDS matrices. Left-circulant matrices are a particular type of Hankel matrices, and thus, from Remark~\ref{Remark_left_circulant_involutory_NMDS}, we can say that Hankel involutory (orthogonal) NMDS matrices of any order may exist over the field $\mathbb{F}_{2^r}$.
\end{remark}

\noindent We close this section by presenting Table~\ref{Table_Comparision_MDS_NMDS}, which compares the involutory and orthogonal properties of MDS and NMDS matrices constructed from the circulant, left-circulant, Toeplitz and Hankel families.

\begin{table}
	\begin{center}
		\caption{Comparison of involutory and orthogonal properties of MDS and NMDS matrices over a finite field $\mathbb{F}_{2^r}$~(``\textbf{DNE}'' stands for does not exist).	\label{Table_Comparision_MDS_NMDS}}
		\begin{tabular}{|p{2cm}|p{2cm}|p{3cm}|p{2cm}|p{2cm}|}
			\hline
			{Type} & {Property} & {Dimension} & {MDS}  & {NMDS} \\ \hline \hline
			\multirow{4}{*}{Circulant} & \multirow{1}{*}{Involutory} &  $n \times n$ & DNE  & DNE \\ \cline{2-5}
			& \multirow{3}{*}{Orthogonal} &  $2^n \times 2^n$ & DNE  &  may exist \\ \cline{3-5}
			& & $2n \times 2n$    & may exist  & may exist  \\ \cline{3-5}
			& & $(2n+1) \times (2n+1)$ & may exist & may exist\\ \hline \hline
			\multirow{3}{*}{left-Circulant} & \multirow{3}{*}{Involutory} & $2^n \times 2^n$ & DNE  & may exist\\
			\cline{3-5}
			& & $2n \times 2n$    & may exist & may exist \\ \cline{3-5}
			& & $(2n+1) \times (2n+1)$ & may exist & may exist \\ \hline \hline
			\multirow{4}{*}{Toeplitz} & \multirow{1}{*}{Involutory} &  $n \times n$ & DNE  & DNE \\ \cline{2-5}
			& \multirow{3}{*}{Orthogonal} &  $2^n \times 2^n$ & DNE  &  may exist  \\ \cline{3-5}
			& & $2n \times 2n$    & may exist  & may exist  \\ \cline{3-5}
			& & $(2n+1) \times (2n+1)$ & may exist & may exist\\ \hline \hline
			\multirow{3}{*}{Hankel} & \multirow{3}{*}{Involutory} & $2^n \times 2^n$ & DNE  & may exist\\
			\cline{3-5}
			& & $2n \times 2n$    & may exist & may exist \\ \cline{3-5}
			& & $(2n+1) \times (2n+1)$ & may exist & may exist \\ \hline
		\end{tabular}
	\end{center}
\end{table}

\section{Construction of nonrecursive NMDS matrices from GDLS matrices}\label{Section_nonrecursive_NMDS_GDLS}

Constructing NMDS matrices from circulant, left-circulant, Hadamard, Toeplitz or Hankel matrices of order $n$ may result in a high implementation cost due to the requirement of having $\mathcal{K}\geq n(n-2)$. To address this issue, in this section, we present some lightweight nonrecursive NMDS matrices through the composition of various GDLS (see Definition~\ref{def_GDLS_matrix_NMDS}) matrices, similar to the method used by the authors in \cite{FUTURE,SM19} for constructing MDS matrices.

To minimize the search space, in most cases, we arbitrarily select $\rho_1$ as the $n$-cycle $[\splitatcommas{n, 1, 2, \ldots, n-1}]$. However, it is important to note that there is no inherent advantage in choosing $\rho_1=[n, 1, 2, \ldots, n-1]$ for obtaining an NMDS matrix. If we change $\rho_1=[n, 1, 2, \ldots, n-1]$ to any permutation from $S_n$, there is still a possibility of obtaining an NMDS matrix.

To search for lightweight nonrecursive NMDS matrices, we examine GDLS matrices of order $n$ with $\mathcal{K}=\ceil{\frac{n}{2}}$ and entries from the set $\set{1,\alpha,\alpha^{-1},\alpha^{2},\alpha^{-2}}$, where $\alpha$ is a primitive element and a root of the constructing polynomial of the field $\FF_{2^r}$. The search space for finding nonrecursive NMDS matrices of order $n \geq 5$ remains large, even when considering the set $\{\splitatcommas{1,\alpha,\alpha^{-1},\alpha^{2},\alpha^{-2}}\}$. Therefore, to obtain nonrecursive NMDS matrices of order $n=5,6,7,8$, we conduct a random search. In addition, to construct nonrecursive NMDS matrices of order $n$, we typically chose $n-2$ GDLS matrices of the same structure. If this does not yield result, we use $n-1$ matrices instead. 

Also note that the implementation costs of the matrices presented in this section over a field are calculated by referring to the s-XOR count value of the corresponding field elements as provided in table of~\cite[App. B]{DSI}.
In Table~\ref{table_comparison_nonrec_NMDS}, we compare our results for nonrecursive NMDS matrices with the existing results.

\begin{table}[h!]
	\centering
	\caption{Comparison of nonrecursive NMDS matrices of order $n$.}\label{table_comparison_nonrec_NMDS}
	\begin{tabular}{|ccccc|}
	\hline
	Order $n$~~ & Input &~~~ Field/Ring  &~~ XOR count 	&~~ References \\ \hline
		$4$	& 4-bit & $\mathbb{F}_{2^4}$ & 24 	& \cite{SM19} \\
		$4$ & 4-bit & $\mathbb{F}_{2^4}$ & 24   & Section~\ref{Section_4_nonrec_NMDS_GDLS}\\
		$4$	& 8-bit & $\mathbb{F}_{2^8}$ & 48 	& \cite{SM19} \\
		$4$ & 8-bit & $\mathbb{F}_{2^8}$ & 48   & Section~\ref{Section_4_nonrec_NMDS_GDLS}\\  \hline
		$5$	& 4-bit & $\FF_{2^4}/0$x$13$ & 65 	& \cite{Li_Wang_2017} \\
		$5$ & 4-bit & $\FF_{2^4}/0$x$13$ &  \textbf{50}  & Section~\ref{Section_5_nonrec_NMDS_GDLS}\\
		$5$	& 8-bit & $\FF_{2^8}/0$x$11$b & 130	& \cite{Li_Wang_2017} \\
		$5$ & 8-bit & $GL(8,\FF_{2})$ &  \textbf{98}  & Remark~\ref{Example_5_nonrec_NMDS_ring_F_2^8}\\ \hline
		$6$	& 4-bit & $\FF_{2^4}/0$x$13$ & 108 	& \cite{Li_Wang_2017} \\
		$6$ & 4-bit & $\FF_{2^4}/0$x$13$ &  \textbf{65}  & Section~\ref{Section_6_nonrec_NMDS_GDLS}\\
		$6$	& 8-bit & $\FF_{2^8}/0$x$11$b & 216	& \cite{Li_Wang_2017} \\
		$6$ & 8-bit & $GL(8,\FF_{2})$ &  \textbf{125}  & Remark~\ref{Example_6_nonrec_NMDS_ring_F_2^8}\\ \hline
		$7$	& 4-bit & $\FF_{2^4}/0$x$13$ & 154 	& \cite{Li_Wang_2017} \\
		$7$ & 4-bit & $\FF_{2^4}/0$x$13$ &  \textbf{96}  & Section~\ref{Section_7_nonrec_NMDS_GDLS}\\
		$7$	& 8-bit & $\FF_{2^8}/0$x$11$b & 308	& \cite{Li_Wang_2017} \\
		$7$ & 8-bit & $GL(8,\FF_{2})$ &  \textbf{176}  & Remark~\ref{Example_7_nonrec_NMDS_ring_F_2^8}\\ \hline
		$8$	& 4-bit & $\FF_{2^4}/0$x$13$ & 216 	& \cite{Li_Wang_2017} \\
		$8$	& 4-bit & $\FF_{2^4}/0$x$13$ & 108 	& \cite{SM19} \\
		$8$ & 4-bit & $\FF_{2^4}/0$x$13$ &  108  & Section~\ref{Section_8_nonrec_NMDS_GDLS}\\
		$8$	& 8-bit & $\FF_{2^8}/0$x$11$b & 432	& \cite{Li_Wang_2017} \\
		$8$	& 8-bit & $GL(8,\FF_{2})$  & 204 	& \cite{SM19} \\
		$8$ & 8-bit & $GL(8,\FF_{2})$ &  204  & Remark~\ref{Example_8_nonrec_NMDS_ring_F_2^8}\\ \hline
	\end{tabular}
\end{table}

\subsection{Construction of $4\times 4$ nonrecursive NMDS matrices}\label{Section_4_nonrec_NMDS_GDLS}
From Remark~\ref{Remark_lowest_XOR_4_rec_NMDS}, we know that the matrix $B$ given in~\ref{Eqn_4-nMDS_over_field_F_2^r} has the lowest XOR count among all $k$-NMDS matrices with $k\leq 4$ over $\mathbb{F}_{2^r}$. 
The proposed GDLS matrix 
\begin{equation*}
	\begin{aligned}
		B&=
		\begin{bmatrix}
			0 & 0 & 0 & 1 \\
			1 & 1 & 0 & 0 \\
			0 & 1 & 0 & 0 \\
			0 & 0 & 1 & 1
		\end{bmatrix}
	\end{aligned}
\end{equation*}
is $3$-NMDS over a field $\mathbb{F}_{2^r}$.
Therefore, we can obtain a nonrecursive NMDS matrix of order $4$ by composing the matrix $B$ with itself three times. This results in an implementation cost of $3\cdot (2\cdot r)=6r$ over a field $\mathbb{F}_{2^r}$.

\subsection{Construction of $5\times 5$ nonrecursive NMDS matrices}\label{Section_5_nonrec_NMDS_GDLS}
In this section, we propose three GDLS matrices, $B_1,B_2$ and $B_3$, which are constructed by the permutations $\rho_1= [5, 1, 2, 3, 4]$, $\rho_2=[3, 2, 5, 4, 1]$ and the following diagonal matrices.
\begin{enumerate}[(i)]
	\item $B_1$:~$\rho_1, \rho_2,~ D_1=diag(1, 1, 1, 1, 1)$ and $D_2=diag(0, \alpha, 0, 1, 1)$
	\item $B_2$:~$\rho_1, \rho_2,~D_1=diag(1, 1, 1, 1, 1)$ and $D_2=diag(0, 1, 0, 1, 1)$
	\item $B_3$:~$\rho_1, \rho_2,~D_1=diag(1, 1, 1, 1, 1)$ and $D_2=diag(0, 1, 0, \alpha, 1),$
\end{enumerate}
where $\alpha$ is a primitive element of~$\FF_{2^4}$ and a root of $x^4+x+1$.
Using these three GDLS matrices, we propose a $5\times 5$ NMDS matrix as follows:
\begin{equation}\label{Example_5_nonrec_NMDS}
    \begin{aligned}
        M = B_2B_3B_1B_2 &= 
        \begin{bmatrix}
			0 & 1 & 0 & 0 & 1 \\
			0 & 1 & 1 & 0 & 0 \\
			0 & 0 & 0 & 1 & 0 \\
			0 & 0 & 0 & 1 & 1 \\
			1 & 0 & 0 & 0 & 0
		\end{bmatrix}
		\begin{bmatrix}
			0 & 1 & 0 & 0 & 1 \\
			0 & 1 & 1 & 0 & 0 \\
			0 & 0 & 0 & 1 & 0 \\
			0 & 0 & 0 & \alpha & 1 \\
			1 & 0 & 0 & 0 & 0
		\end{bmatrix}
		\begin{bmatrix}
			0 & 1 & 0 & 0 & 1 \\
			0 & \alpha & 1 & 0 & 0 \\
			0 & 0 & 0 & 1 & 0 \\
			0 & 0 & 0 & 1 & 1 \\
			1 & 0 & 0 & 0 & 0
		\end{bmatrix}
		\begin{bmatrix}
			0 & 1 & 0 & 0 & 1 \\
			0 & 1 & 1 & 0 & 0 \\
			0 & 0 & 0 & 1 & 0 \\
			0 & 0 & 0 & 1 & 1 \\
			1 & 0 & 0 & 0 & 0
		\end{bmatrix}\\
		&= 
		\begin{bmatrix}
		    1 & \alpha + 1 & \alpha + 1 & 0 & 1 \\
            1 & \alpha & \alpha & 1 & 0 \\
            \alpha & 1 & 0 & \alpha & \alpha + 1 \\
            \alpha + 1 & 0 & 1 & \alpha & \alpha + 1 \\
            0 & \alpha + 1 & \alpha & 1 & 1
		\end{bmatrix}.
    \end{aligned}
\end{equation}
Now, $XOR(M)=XOR(B_1)+2\cdot XOR(B_2)+ XOR(B_3)$. Therefore, $M$ can be implemented with $(1+3\cdot 4)+2\cdot (0+3\cdot 4)+ (1+3\cdot 4)=50$ XORs over the field $\FF_{2^4}/0$x$13$.

\begin{remark}\label{Example_5_nonrec_NMDS_ring_F_2^8}
    As discussed in Remark~\ref{Example_k-nMDS_order_5_6_ring_F_2^8}, if $\alpha$ is replaced with $C$, the matrix $M$ from (\ref{Example_5_nonrec_NMDS}) will be NMDS over $GL(8,\FF_{2})$, with an implementation cost of $(1+3\cdot 8)+2\cdot (0+3\cdot 8)+ (1+3\cdot 8)=98$ XORs.
\end{remark}

\subsection{Construction of $6\times 6$ nonrecursive NMDS matrices}\label{Section_6_nonrec_NMDS_GDLS}
In this section, we propose a lightweight $6\times 6$ NMDS matrix $M$ that can be implemented with $65$~XORs over the field $\mathbb{F}_{2^4}$. The matrix $M$ is constructed from three GDLS matrices, $B_1$, $B_2$, and $B_3$, of order $6$, as $M=B_2^2B_1B_3B_2$. These GDLS matrices are constructed using the permutations $\rho_1 = [6, 1, 2, 3, 4, 5]$ and $\rho_2 = [5, 6, 1, 2, 3, 4]$ and by the following diagonal matrices as follows:

\begin{enumerate}[(i)]
	\item $B_1$:~$\rho_1, \rho_2,~D_1=diag(1, 1, 1, \alpha, 1, \alpha)$ and $D_2=diag(0, \alpha, 0, 1, 0, 1)$
	\item $B_2$:~$\rho_1, \rho_2,~D_1=diag(1, 1, 1, 1, 1,1)$ and $D_2=diag(0, 1, 0, 1, 0, 1)$
	\item $B_3$:~$\rho_1, \rho_2,~D_1=diag(\alpha, 1, 1, \alpha^{-1}, 1, 1)$ and $D_2=diag(0, 1, 0, 1, 0, 1)$
\end{enumerate}

\begin{equation}\label{Example_6_nonrec_NMDS}
	\begin{aligned}
		B_1=&\begin{bmatrix}
			0 & 1 & 0 & 0 & 0 & 0 \\
			0 & 0 & 1 & 1 & 0 & 0 \\
			0 & 0 & 0 & \alpha & 0 & 0 \\
			0 & 0 & 0 & 0 & 1 & 1 \\
			0 & 0 & 0 & 0 & 0 & \alpha \\
			1 & \alpha & 0 & 0 & 0 & 0
		\end{bmatrix}~~
		B_2=
		\begin{bmatrix}
			0 & 1 & 0 & 0 & 0 & 0 \\
			0 & 0 & 1 & 1 & 0 & 0 \\
			0 & 0 & 0 & 1 & 0 & 0 \\
			0 & 0 & 0 & 0 & 1 & 1 \\
			0 & 0 & 0 & 0 & 0 & 1 \\
			1 & 1 & 0 & 0 & 0 & 0
		\end{bmatrix}~~
		B_3=
		\begin{bmatrix}
			0 & 1 & 0 & 0 & 0 & 0 \\
			0 & 0 & 1 & 1 & 0 & 0 \\
			0 & 0 & 0 & \alpha^{-1} & 0 & 0 \\
			0 & 0 & 0 & 0 & 1 & 1 \\
			0 & 0 & 0 & 0 & 0 & 1 \\
			\alpha & 1 & 0 & 0 & 0 & 0
		\end{bmatrix},
	\end{aligned}
\end{equation}
where $\alpha$ is a primitive element of~$\FF_{2^4}$ and a root of $x^4+x+1$. Now it can be checked that $M$ is NMDS over $\FF_{2^4}/0$x$13$ with an implementation cost of $65$ XORs, calculated as $XOR(M)=XOR(B_1)+3\cdot XOR(B_2)+XOR(B_3)=(1+1+1+3\cdot 4)+ 3\cdot (0+3\cdot 4)+(1+1+3\cdot 4)=65$.

\begin{remark}\label{Example_6_nonrec_NMDS_ring_F_2^8}
    As discussed in Remark~\ref{Example_k-nMDS_order_5_6_ring_F_2^8}, if $\alpha$ is replaced with $C$, the matrix $M$ constructed from $B_1,B_2$ and $B_3$ in (\ref{Example_6_nonrec_NMDS}) will be NMDS over $GL(8,\FF_{2})$, with an implementation cost of $125$ XORs.
\end{remark}

\subsection{Construction of $7\times 7$ nonrecursive NMDS matrices}\label{Section_7_nonrec_NMDS_GDLS}
This section presents three GDLS matrices, $B_1, B_2$, and $B_3$, of order $7$. These matrices are constructed using the permutations $\rho_1=[6,7,4,5,2,3,1]$ and $\rho_2=[3,2,1,4,7,6,5]$, along with specific diagonal matrices as follows:
\begin{enumerate}[(i)]
	\item $B_1$:~$\rho_1, \rho_2,~D_1=diag(1, \alpha^{-1}, 1, \alpha^{-2}, 1, \alpha^2, 1)$ and $D_2=diag(1, 0, 1, 0, 1, 0, 1)$
	\item $B_2$:~$\rho_1, \rho_2,~D_1=diag(1, 1, 1, 1, 1,1,1)$ and $D_2=diag(1, 0, 1, 0, 1, 0, 1)$
	\item $B_3$:~$\rho_1, \rho_2,~D_1=diag(1, 1, 1, 1, 1, 1, \alpha^{-1})$ and $D_2=diag(1, 0, 1, 0, \alpha^{-2}, 0, 1)$
\end{enumerate}

\begin{equation}\label{Eqn_Example_nonrec_7_NMDS}
	\begin{aligned}
		B_1&=
		\begin{bmatrix}
			0 & 0 & 1 & 0 & 0 & 0 & 1 \\
			0 & 0 & 0 & 0 & 1 & 0 & 0 \\
			1 & 0 & 0 & 0 & 0 & \alpha^{2} & 0 \\
			0 & 0 & 1 & 0 & 0 & 0 & 0 \\
			0 & 0 & 0 & \alpha^{-2} & 0 & 0 & 1 \\
			1 & 0 & 0 & 0 & 0 & 0 & 0 \\
			0 & \alpha^{-1} & 0 & 0 & 1 & 0 & 0
		\end{bmatrix}~~
		B_2=
		\begin{bmatrix}
			0 & 0 & 1 & 0 & 0 & 0 & 1 \\
			0 & 0 & 0 & 0 & 1 & 0 & 0 \\
			1 & 0 & 0 & 0 & 0 & 1 & 0 \\
			0 & 0 & 1 & 0 & 0 & 0 & 0 \\
			0 & 0 & 0 & 1 & 0 & 0 & 1 \\
			1 & 0 & 0 & 0 & 0 & 0 & 0 \\
			0 & 1 & 0 & 0 & 1 & 0 & 0
		\end{bmatrix}~~
		B_3=
		\begin{bmatrix}
			0 & 0 & 1 & 0 & 0 & 0 & \alpha^{-1} \\
			0 & 0 & 0 & 0 & 1 & 0 & 0 \\
			1 & 0 & 0 & 0 & 0 & 1 & 0 \\
			0 & 0 & 1 & 0 & 0 & 0 & 0 \\
			0 & 0 & 0 & 1 & 0 & 0 & 1 \\
			1 & 0 & 0 & 0 & 0 & 0 & 0 \\
			0 & 1 & 0 & 0 & \alpha^{-2} & 0 & 0
		\end{bmatrix},
	\end{aligned}
\end{equation}
where $\alpha$ is a primitive element in $\FF_{2^4}$ and a root of the polynomial $x^4+x+1$.
Using these three GDLS matrices, we propose a $7\times 7$ matrix $M$ given by $M=B_3B_1^2B_3B_2$. It can be verified that $M$ is an NMDS matrix over $\FF_{2^4}/0$x$13$ with an implementation cost of $96$ XORs, which is calculated as $XOR(M)=~2\cdot(1+2+2+4\cdot4)+(0+4\cdot4)+2\cdot(1+2+4\cdot4)=96$.

\begin{remark}\label{Example_7_nonrec_NMDS_ring_F_2^8}
    By replacing $\alpha$ with $C$, as discussed in Remark~\ref{Example_k-nMDS_order_5_6_ring_F_2^8}, the matrix $M$ constructed from $B_1,B_2$ and $B_3$ in (\ref{Eqn_Example_nonrec_7_NMDS}) becomes an NMDS over $GL(8,\FF_{2})$. Furthermore, the binary matrices $C^2$ and $C^{-2}$ can be implemented with $2$ XORs. Consequently, the implementation cost of the matrix $M$ is $176$ XORs over $GL(8,\FF_{2})$.
\end{remark}

\subsection{Construction of $8\times 8$ nonrecursive NMDS matrices}\label{Section_8_nonrec_NMDS_GDLS}
In this section, we present a lightweight $8\times 8$ matrix $M$ over the field $\mathbb{F}_{2^4}$ that can be implemented with $108$ XORs, which meets the best known result.
To construct the matrix $M$, we use three GDLS matrices, $B_1$, $B_2$, and $B_3$, of order $8$, by $M=B_2 B_1 B_3 B_2^3$. These GDLS matrices are generated using the permutations $\rho_1 = [4, 5, 2, 3, 8, 1, 6, 7]$ and $\rho_2 = [5, 4, 3, 6, 1, 8, 7, 2]$, along with the following diagonal matrices.

\begin{enumerate}[(i)]
	\item $B_1$:~$\rho_1, \rho_2,D_1=diag(1, \alpha, 1, \alpha,~ 1, \alpha, 1, \alpha)$ and 
	$D_2=diag(\splitatcommas{1, 0, 1, 0,~ 1, 0, 1, 0})$
	\item $B_2$:~$\rho_1, \rho_2,D_1=diag(1, 1, 1, 1,~ 1, 1, 1, 1)$, $D_2=diag(1, 0, 1, 0,~ 1, 0, 1, 0)$
	\item $B_3$:~$\rho_1, \rho_2,D_1=diag(1, 1, 1, 1,~ 1, 1, 1, 1)$ and  $D_2=diag(\splitatcommas{\alpha^{-2}, 0, \alpha^{-2}, 0,~ \alpha^{-2}, 0, \alpha^{-2},0})$
\end{enumerate}
\begin{equation}\label{Eqn_Example_nonrecursive_8_NMDS}
	\begin{aligned}
		B_1&{=}
		\begin{bmatrix}
			0 & 0 & 0 & 0 & 1 & \alpha & 0 & 0 \\
			0 & 0 & 1 & 0 & 0 & 0 & 0 & 0 \\
			0 & 0 & 1 & \alpha & 0 & 0 & 0 & 0 \\
			1 & 0 & 0 & 0 & 0 & 0 & 0 & 0 \\
			1 & \alpha & 0 & 0 & 0 & 0 & 0 & 0 \\
			0 & 0 & 0 & 0 & 0 & 0 & 1 & 0 \\
			0 & 0 & 0 & 0 & 0 & 0 & 1 & \alpha \\
			0 & 0 & 0 & 0 & 1 & 0 & 0 & 0
		\end{bmatrix}
		B_2{=}
		\begin{bmatrix}
			0 & 0 & 0 & 0 & 1 & 1 & 0 & 0 \\
			0 & 0 & 1 & 0 & 0 & 0 & 0 & 0 \\
			0 & 0 & 1 & 1 & 0 & 0 & 0 & 0 \\
			1 & 0 & 0 & 0 & 0 & 0 & 0 & 0 \\
			1 & 1 & 0 & 0 & 0 & 0 & 0 & 0 \\
			0 & 0 & 0 & 0 & 0 & 0 & 1 & 0 \\
			0 & 0 & 0 & 0 & 0 & 0 & 1 & 1 \\
			0 & 0 & 0 & 0 & 1 & 0 & 0 & 0
		\end{bmatrix}
		B_3{=}
		\begin{bmatrix}
		    0 & 0 & 0 & 0 & \alpha^{-2} & 1 & 0 & 0 \\
			0 & 0 & 1 & 0 & 0 & 0 & 0 & 0 \\
			0 & 0 & \alpha^{-2} & 1 & 0 & 0 & 0 & 0 \\
			1 & 0 & 0 & 0 & 0 & 0 & 0 & 0 \\
			\alpha^{-2} & 1 & 0 & 0 & 0 & 0 & 0 & 0 \\
			0 & 0 & 0 & 0 & 0 & 0 & 1 & 0 \\
			0 & 0 & 0 & 0 & 0 & 0 & \alpha^{-2} & 1 \\
			0 & 0 & 0 & 0 & 1 & 0 & 0 & 0
		\end{bmatrix},
	\end{aligned}
\end{equation}
where $\alpha$ is a primitive element of~$\FF_{2^4}$ and a root of $x^4+x+1$. Therefore, $M$ can be implemented with $(1+1+1+1+4\cdot 4)+4\cdot (0+4\cdot 4) + (2+2+2+2+4\cdot 4)=108$ XORs.

\begin{remark}\label{Example_8_nonrec_NMDS_ring_F_2^8}
    As discussed in Remark~\ref{Example_k-nMDS_order_5_6_ring_F_2^8}, if we substitute $\alpha$ with $C$, the matrix $M$ that is formed from $B_1, B_2,$ and $B_3$ in (\ref{Eqn_Example_nonrecursive_8_NMDS}) becomes an NMDS matrix over $GL(8,\FF_{2})$. Also, the binary matrix $C^{-2}$ can be implemented with only $2$ XORs. Therefore, the implementation cost of the matrix $M$ becomes $204$ XORs over $GL(8,\FF_{2})$.
\end{remark}

\begin{remark}
	From Table~\ref{table_comparison_nonrec_NMDS}, we can observe that the proposed nonrecursive NMDS matrices of order $4$ and $8$ have the same cost as given in the paper~\cite{SM19} for $\mathbb{F}_{2^4}$ and $\mathbb{F}_{2^8}$. However, it should be noted that the paper~\cite{SM19} does not provide NMDS matrices for the orders $n=5$, $6$, and $7$. In contrast, our proposed nonrecursive NMDS matrices for orders 5, 6, and 7 not only fulfill this gap but also have a lower hardware cost compared to the existing nonrecursive NMDS matrices in the literature~\cite{Li_Wang_2017} (as shown in Table~\ref{table_comparison_nonrec_NMDS}). Also, the GDLS matrix structure is not limited to even orders, unlike the GFS or EGFS matrix structure used in~\cite{SM19}. The GDLS matrix structure is applicable to matrices of all orders, enabling improvements in several parameters that are not achievable with the GFS or EGFS matrix structure.
\end{remark}

\begin{table}[h!]
	\centering
	\caption{A summary of results on NMDS matrices of this paper.}\label{table_comparison_our_results}
	\begin{tabular}{|ccccc|}
	\hline
	Order $n$~~ & Input &~~ Type &~~ Iterations &~~ XOR count \\ \hline
	$4$	& 4-bit & recursive & 3 & 8\\
	$4$	& 4-bit & nonrecursive  & - & 24\\
	$4$	& 8-bit & recursive & 3 & 16\\
	$4$	& 8-bit & nonrecursive  & - & 48\\ \hline
	$5$	& 4-bit & recursive & 4 & 14\\
	$5$	& 4-bit & recursive & 5 & 13\\
	$5$	& 4-bit & nonrecursive  & - & 50\\
	$5$	& 8-bit & recursive & 4 & 26\\
	$5$	& 8-bit & recursive & 5 & 25\\
	$5$	& 8-bit & nonrecursive  & - & 98\\
	\hline
	$6$	& 4-bit & recursive & 5 & 14\\
	$6$	& 4-bit & recursive & 6 & 13\\
	$6$	& 4-bit & nonrecursive  & - & 65\\
	$6$	& 8-bit & recursive & 4 & 26\\
	$6$	& 8-bit & recursive & 5 & 25\\
	$6$	& 8-bit & nonrecursive  & - & 125\\
	\hline
	$7$	& 4-bit & recursive & 5 & 22\\
	$7$	& 4-bit & recursive & 6 & 20\\
	$7$	& 4-bit & recursive & 7 & 18\\
	$7$	& 4-bit & nonrecursive  & - & 96\\
	$7$	& 8-bit & recursive & 5 & 38\\
	$7$	& 8-bit & recursive & 6 & 36\\
	$7$	& 8-bit & recursive & 7 & 34\\
	$7$	& 8-bit & nonrecursive  & - & 176\\
	\hline
	$8$	& 4-bit & nonrecursive  & - & 108\\
	$8$	& 8-bit & recursive & 7 & 40\\
	$8$	& 8-bit & recursive & 8 & 38\\
	$8$	& 8-bit & nonrecursive  & - & 204\\ \hline
	\end{tabular}
\end{table}

\section{Conclusion and Future Work}\label{Section_NMDS_conclusion_future_work}
This paper examines the construction of NMDS matrices using both recursive and nonrecursive approaches, investigates various theoretical results, and presents some lightweight NMDS matrices of various orders in both approaches.
Table~\ref{table_comparison_our_results} presents a summary of the implementation cost of the NMDS matrices given in this paper. We explore the DLS matrices and derive some theoretical results for the construction of recursive NMDS matrices.
We prove that for $n\geq 4$, there does not exist any $k$-NMDS sparse matrix of order $n$ with $\mathcal{K}=1$ and $k\leq n$ over a field of characteristic $2$.
For the nonrecursive NMDS matrices, we examine the circulant, left-circulant, Toeplitz, and Hankel families of matrices. We prove that Toeplitz matrices of order $n> 4$ cannot be simultaneously NMDS and involutory over a field of characteristic 2.
To compare with MDS matrices, we examine some well-known results of MDS matrices and apply them to NMDS matrices. For instance, Table~\ref{Table_Comparision_MDS_NMDS} compares the involutory and orthogonal properties of MDS and NMDS matrices constructed from circulant, left-circulant, Toeplitz, and Hankel matrices.
We use GDLS matrices to provide some lightweight NMDS matrices that can be computed in one clock cycle. The proposed nonrecursive NMDS matrices of orders 4, 5, 6, 7, and 8 can be implemented with 24, 50, 65, 96, and 108 XORs over $\mathbb{F}_{2^4}$, respectively. These results match the best-known lightweight NMDS matrices of orders 4 and 8, and outperform the best-known matrices of orders 5, 6, and 7.

\par In the literature, there has been an extensive study of the direct construction of MDS matrices using both recursive and nonrecursive methods. Nonrecursive direct constructions are mainly obtained from Cauchy and Vandermonde based constructions, while recursive direct constructions are derived from companion matrices by some coding theoretic techniques. However, there is no direct construction available for the NMDS matrices. Therefore, finding a direct construction method for NMDS matrices using both recursive and nonrecursive approaches could be a potential area for future work.

\medskip

\bibliographystyle{plain}
\bibliography{full_NMDS}
\end{document}